\documentclass{style/vldb}

\usepackage{algorithm}
\usepackage{algorithmic}
\usepackage{color}
\usepackage{graphicx}
\usepackage{array}
\usepackage{caption}
\usepackage{subcaption}
\usepackage{multirow}
\usepackage{hyperref}

\usepackage{amsthm}

\newcommand{\hindent}[1][1]{\hspace{#1\algorithmicindent}}

\newtheorem{theorem}{Theorem}
\newtheorem{lemma}{Lemma}
\newtheorem{exmp}{Example}
 
\theoremstyle{definition}
\newtheorem{definition}{Definition}

\usepackage[font={small,bf}]{caption} 
\begin{document}

\title{Efficient Computation of Subspace Skyline over Categorical Domains}

\numberofauthors{1}
\author{
  \alignauthor
  Md Farhadur Rahman$^\dag$, Abolfazl Asudeh$^\dag$, Nick Koudas$^\ddag$, Gautam Das$^\dag$\\
  \affaddr { University of Texas at Arlington$^\dag$; University of Toronto$^\ddag$}
  {
    \email{
        $^{\dag}$\{mdfarhadur.rahman@mavs,~ab.asudeh@mavs,~gdas@cse\}.uta.edu, $^{\ddag}$koudas@cs.toronto.edu
    }
  }
}

\date{}

\maketitle

\begin{abstract}
Platforms such as AirBnB, Zillow, Yelp, and related sites have transformed the way we search for accommodation, restaurants, etc.
The underlying datasets in such applications have numerous attributes that are mostly Boolean or Categorical. 
Discovering the skyline of such datasets over a subset of attributes would identify entries that stand out while enabling numerous applications.
There are only a few algorithms designed to compute the skyline over categorical attributes, yet are applicable only when the number of attributes is small.

In this paper, we place the problem of skyline discovery over categorical attributes into perspective and design efficient algorithms for two cases.
(i) In the absence of indices, we propose two algorithms, ST-S and ST-P, that exploit the categorical characteristics of the datasets, organizing 
tuples in a tree data structure, supporting efficient dominance tests over the candidate set.
(ii) We then consider the existence of widely used precomputed sorted lists.
After discussing several approaches, and studying their limitations, we propose TA-SKY, a novel threshold style algorithm that utilizes sorted lists. 
Moreover, we further optimize TA-SKY and explore its progressive nature, making it suitable for applications with strict interactive requirements.
In addition to the extensive theoretical analysis of the proposed algorithms, we conduct a comprehensive experimental evaluation of the combination of real (including the entire AirBnB data collection) and synthetic datasets to study the practicality of the proposed algorithms. The results showcase the superior performance of our techniques, outperforming applicable approaches by orders of magnitude.
\end{abstract}

\section{Introduction}\label{sec:inro}
\subsection{Motivation}
\label{sec:motivation}
Skyline queries are widely used in applications involving multi-criteria decision making~\cite{hwang2012multiple}, and are further related to well-known problems such as top-$k$ queries~\cite{ilyas2008survey,asudeh2016discovering}, preference collection~\cite{asudeh2015crowdsourcing}, and nearest neighbor search~\cite{kossmann2002}.
Given a set of tuples, skylines are computed by considering the dominance relationships among them. A tuple $p$ \textit{dominates} another tuple $q$, if $q$ is not better than $p$ in any dimension and $p$ is better than $q$ in at least one dimension. Moreover, a pair of tuples $p$ and $q$ are considered to be \textit{incomparable} if neither $p$ nor $q$ dominates the other. The {\em Skyline} is the set of tuples that are not dominated by any other tuple in the dataset~\cite{borzsony2001skyline}.

In recent years, several applications have gained popularity in assisting users in tasks ranging from selecting a hotel in an area to locating a nearby restaurant. AirBnB, TripAdvisor, hotels.com, Craigslist, and Zillow are a few such examples. The underlying datasets have numerous attributes that are mostly Boolean or categorical. 
They also include a few numeric attributes, but in most cases the numeric attributes are discretized and transformed into categorical attributes~\cite{morse2007efficient}.
For example, in the popular accommodation rental service AirBnB, the typical attributes are type and number of rooms, types of amenities offered, the number of occupants, etc. Table~\ref{tab:hostDataset} shows a toy example that contains a subset of attributes present in AirBnB. Note that most of the attributes are amenities provided by the hosts (the temporary rental providers) and are primarily Boolean. The AirBnB dataset features more than 40 such attributes describing amenities users can choose. One way of identifying desirable hosts in such a dataset is to focus on the non-dominated hosts. This is because if a listing $t$ dominates another listing $t'$ (i.e., $t$ is at least as good as $t'$ on all the attributes while is better on at least one attribute), $t$ should naturally be preferred over $t'$.

In the example shown in Table~\ref{tab:hostDataset}, "Host 1" and "Host 2" are in the skyline, while all the others are dominated by at least one of them.
In real-world applications, especially when the number of attributes increases, users naturally tend to focus on a subset of attributes that is of interest to them. For example, during an AirBnB query, 
we typically consider a few attributes while searching for hosts that are in the skyline. For instance, in the dataset shown in Table~\ref{tab:hostDataset}, one user might be interested in \textit{Breakfast} and \textit{Internet}, while another user might focus on \textit{Internet}, \textit{Cable TV}, and \textit{Pool} when searching for a host.

\begin{table}[!t]
\centering
\caption{A sample hosts dataset}\label{tab:hostDataset}
\begin{tabular}{|p{1cm}|p{1.3cm}|p{0.70cm}|p{0.7cm}|p{1cm}|p{1cm}|}
    \hline 
    Host Name & Breakfast & Pool & Cable TV & Internet & Ratings\\
    \hline 
    Host 1 & T & F & T & T & 4.0\\
    Host 2 & T & T & F & T & 4.5\\
    Host 3 & T & F & F & T & 3.5\\
    Host 4 & T & F & F & F & 3.0\\
    Host 5 & F & F & T & T & 3.5\\
    \hline
\end{tabular}
\end{table}

In this paper, we consider the problem of {\em subspace skyline discovery} over such datasets, in which given an ad-hoc subset of attributes as a query, the goal is to identify 
the tuples in the skyline involving only those attributes\footnote{Naturally this definition includes skyline discovery over all attributes of a relation.}. Such subspace skyline queries are an effective tool in assisting users in data exploration (e.g., an AirBnB customer can explore the returned skyline to narrow down to a preferred host). 

In accordance with common practice in traditional database query processing, we design solutions for two important practical instances of this problem, namely: (a) assuming that no indices exist on the underlying dataset, and (b) assuming that indices exist on each individual attribute of the dataset. The space devoted to indices is a practical concern; given that the number of possible subset queries is exponential we do not consider techniques that would construct indices for each possible subset as that would impose an exponential storage overhead (not to mention increased overhead for maintaining such indices under dynamic updates as it is typical in our scenario). Thus we explore a solution space in which index overhead ranges from zero to linear in the number of attributes, trading space for increased performance as numerous techniques in database query processing typically do \cite{gupta1995aggregate, das2006answering, halevy2001answering}.

To the best of our knowledge, LS~\cite{morse2007efficient} and Hexagon~\cite{preisinger2007hexagon} are the only two algorithms designed to compute skylines over categorical attributes. Both of these algorithms operate by creating {\em a lattice} over the attributes in a skyline query, which is feasible only when the number of attributes is really small.

\vspace{-2mm}
\subsection{Technical Highlights}
In this paper, we propose efficient algorithms to effectively identify the answer for any subspace skyline query. Our main focus is to overcome the limitations of previous works (\cite{morse2007efficient, preisinger2007hexagon}), introducing efficient and scalable skyline algorithms for categorical datasets.

For the case when no indices are available, we design a tree structure to arrange the tuples in a ``candidate skyline'' set. 
The tree structure supports efficient dominance tests over the candidate set, thus reducing the overall cost of skyline computation. 
We then propose two novel algorithms called {\bf ST-S} (Skyline using Tree Sorting-based)  and {\bf ST-P} (Skyline using Tree Partition-based) 
that incorporate the tree structure into existing sorting- and partition-based algorithms. Both ST-S and ST-P work when no index is available on the underlying datasets and deliver superior performance for any subset skyline query.

Then, we utilize precomputed sorted lists~\cite{fagin2003optimal} and design efficient algorithms for the index-based version of our problem. 
As one of the main results of our paper, we propose the Threshold Algorithm for Skyline ({\bf TA-SKY}) capable of answering subspace skyline queries. In the context of {\bf TA-SKY}, we first start with a brief discussion of a few approaches that operate by constructing a full/partial lattice over the query space. However, these algorithms have a complexity that is exponential in the number of attributes involved in the skyline query. To overcome this limitation, we propose {\bf TA-SKY}, an interesting adaptation of the top-$K$ threshold (TA)~\cite{fagin2003optimal} style of processing for the subspace skyline problem. TA-SKY utilizes sorted lists and constructs the projection of the tuples in query space. 

TA-SKY proceeds by accumulating information, utilizing sequential access over the indices that enable it to stop early while guaranteeing that all skyline tuples have been identified. The early stopping condition enables TA-SKY to answer skyline queries {\em without accessing all the tuples}, thus reducing the total number of dominance checks, resulting in greater efficiency.
Consequently, as further discussed in \S\ref{sec:experiments}, TA-SKY demonstrates an order of magnitude speedup during our experiments.
In addition to TA-SKY, we subsequently propose novel optimizations to make the algorithm even more efficient. TA-SKY is an online algorithm - it can output a subset of skyline tuples without discovering the entire skyline set. The progressive characteristic of TA-SKY makes it suitable for web applications, with strict interactive requirements, where users want to get a subset of results very quickly.
We study this property of TA-SKY in \S\ref{sec:experiments} on the {\em entire AirBnB} data collection for which TA-SKY discovered more than two-thirds of the skyline in less than $3$ seconds while accessing around $2\%$ of the tuples, demonstrating the practical utility of our proposal.

\subsection{Summary of Contributions}
We propose a comprehensive set of algorithms for the subspace skyline discovery problem over categorical domains.
The summary of main contributions of this paper are as follows:
\begin{itemize}  
    \itemsep0em 
    \item We present a novel tree data structure that supports efficient dominance tests over relations with categorical attributes.
    \item We propose the ST-S and ST-P algorithms that utilize the tree data structure for the subspace skyline discovery problem, in the absence of indices.
    \item We propose TA-SKY, an efficient algorithm for answering subspace skyline queries with a linear worst case cost dependency to the number of attributes. The progressive characteristic of TA-SKY makes it suitable for interactive web-applications. This is a novel and the first (to our knowledge) adaptation of the TA style of processing to a skyline problem.
    \item We present a comprehensive theoretical analysis of the algorithms quantifying their performance analytically, and present the expected cost of each algorithm.
    \item We present the results of extensive experimental evaluations of the proposed algorithms over real-world and synthetic datasets at scale showing the benefits of our proposals. In particular, in all cases considered we demonstrate that the performance benefits of our approach are extremely large (in most cases by orders of magnitude) when compared to other applicable approaches.
\end{itemize}

\subsection{Paper Organization}
The rest of the paper is organized as follows. We discuss preliminaries, notations, and problem definition in \S\ref{sec:preliminaries}. Then, in \S\ref{sec:3}, we present the algorithm for identifying the subspace skyline over low-cardinality datasets, in the absence of precomputed indices. The algorithms for the case of considering the precomputed sorted lists are discussed in \S\ref{sec:subsky}. Following related work in \S\ref{sec:relWork}, we present the experimental results in \S\ref{sec:experiments}. \S\ref{sec:conclusion} concludes the paper.

\section{Preliminaries}\label{sec:preliminaries}
Consider a relation $D$ with $n$ tuples and $m+1$ attributes. 
One of the attributes is $tupleID$, which has a unique value for each tuple.
Let the remaining $m$ categorical attributes be $\mathcal{A}=\{A_1,\dots ,A_m\}$. 
Let $Dom(\cdot)$ be a function that returns the domain of one or more attributes. For example, $Dom(A_i)$ represents the
domain of $A_i$, while $Dom(\mathcal{A})$ represents the Cartesian product
of the domains of attributes in $\mathcal{A}$. $|Dom(A_i)|$ represents the cardinality of $Dom(A_i)$. We use $t[A_i]$ to denote the value of $t$ on the attribute $A_i$.  
We also assume that for each attribute, the values in the domain have a total ordering by preference
(we shall  use overloaded notation such as $a > b$ to indicate that value $a$ is preferred over value $b$).

\subsection{Skyline}
We now define the notions of \textit{dominance} 
and \textit{skyline}~\cite{borzsony2001skyline} formally. 

\begin{definition}{(Dominance).}
A tuple $t\in D$ dominates a tuple $t^\prime\in D$, denoted by $t \succ t^\prime$, {\it iff} $\forall A\in\mathcal{A},\, t[A] \geq t^\prime[A]$ and $\exists A \in \mathcal{A}, \, t[A] > t^\prime[A]$.
Moreover, a tuple $t\in D$ is not comparable with a tuple $t^\prime\in D$, denoted by $t \sim t^\prime$, {\it iff} $t \nsucc t^\prime$ and $t^\prime \nsucc t$.
\end{definition}

\begin{definition}{(Skyline).}
Skyline, $\mathcal{S}$, is the set of tuples  that are not dominated by any other tuples in $D$, i.e.: $\mathcal{S} = \{t\in D|\nexists t^\prime \in D \mbox{ s.t. } t^\prime \succ t\}$
\end{definition}

For each tuple $t \in D$, we shall also be interested in computing its $score$ value, denoted by $score(t)$, using a monotonic function $F(\cdot)$. A function $F(\cdot)$ satisfies the monotonicity condition if $F(t) \geq F(t') \Rightarrow t' \nsucc t$.

\vspace{1mm}
\noindent{\bf Subspace Skyline:} Let $\mathcal{Q} \subseteq \mathcal{A}$ be a subset of attributes. The attributes in $\mathcal{Q}$ forms a $|\mathcal{Q}|$-dimensional subspace of $\mathcal{A}$. The projection of a tuple $t \in D$ in subspace $\mathcal{Q}$ is denoted by $t_{\mathcal{Q}}$ where $t_{\mathcal{Q}}[A] = t[A], \, \forall A \in \mathcal{Q}$. Let $D_{\mathcal{Q}}$ be the projection of all tuples of $D$ in subspace $\mathcal{Q}$ . A tuple $t_{\mathcal{Q}} \in D_{\mathcal{Q}}$ dominates another tuple $t^\prime_{\mathcal{Q}} \in D_{\mathcal{Q}}$ in subspace $\mathcal{Q}$ (denoted by $t_{\mathcal{Q}} \succ_{\mathcal{Q}} t^\prime_{\mathcal{Q}}$) if $t^\prime_\mathcal{Q}$ is not preferred to $t$ on any attribute in $\mathcal{Q}$ while $t$ is preferred to $t^\prime$ on least one attribute in $\mathcal{Q}$.

\begin{definition}{(Subspace Skyline).}
Given a subspace $\mathcal{Q}$, the Subspace Skyline, $\mathcal{S_\mathcal{Q}}$, is the set of tuples in $D_{\mathcal{Q}}$ that are not dominated by any other tuples, i.e.: $\mathcal{S_\mathcal{Q}} = \{t_{\mathcal{Q}} \in D_{\mathcal{Q}} | \nexists t^\prime_{\mathcal{Q}} \in D_{\mathcal{Q}} \mbox{ s.t. } t^\prime_{\mathcal{Q}} \succ_{\mathcal{Q}} t_{\mathcal{Q}}\}$
\end{definition}

\subsection{Sorted Lists}
{\em Sorted lists} are popular data structures widely used by many access-based techniques in data management~\cite{fagin1996combining,fagin2003optimal}.
Let $\mathcal{L} = \{ L_1, L_2, \ldots, L_m \}$ be $m$ sorted lists, where $L_i$ corresponds to a (descending) sorted list for attribute $A_i$. All these lists have the same length, $n$ (i.e., one entry for each tuple in the relation). Each entry of $L_i$ is a pair of the form $(tupleID, t[A_i])$. 


A sorted list supports two modes of access: (i) \textit{sorted (or sequential) access}, and (ii) \textit{random access}. Each call to \textit{sorted access} returns an entry with the next highest attribute value. Performing \textit{sorted access} $k$ times on list $L_i$ will return the first $k$ entries in the list. In \textit{random access} mode, we can retrieve the attribute value of a specific tuple. A \textit{random access} on list $L_i$ assumes $tupleID$ of a tuple $t$ as input and returns the corresponding attribute value $t[A_i]$.

\subsection{Problem Definition}
In this paper, we address the efficient computation of {\em subspace skyline} queries over a relation with categorical attributes.
Formally:

\medskip\noindent
 \framebox[\columnwidth]{\parbox{0.9\columnwidth}{ \textsc{Subspace Skyline Discovery:}
Given
a relation $D$ with the set of categorical attributes $\mathcal{A}$ 
and a subset of attributes in the form of a subspace skyline query $\mathcal{Q}\subseteq \mathcal{A}$,
find
the skyline over $\mathcal{Q}$, denoted by $\mathcal{S}_{\mathcal{Q}}$.
}}\\

In answering subspace skyline queries we consider two scenarios: (i) no precomputed indices are available, and (ii) existence of precomputed sorted lists.

Table~\ref{tab:notations} lists all the notations that are used throughout the paper (we shall introduce some of these later in the paper).
\begin{table}[!t]
\begin{tiny}
\centering
\caption{Table of notations}\label{tab:notations}
\begin{tabular}{|l|p{6cm}|}
    \hline 
    {\bf Notation} & {\bf Semantics}\\
    \hline
    $D$ & Relation\\
    \hline
    $n$ & Number of tuples in the relation\\
    \hline
    $m$ & Number of attributes\\
    \hline
    $t_1, \ldots, t_n$ & Set of tuples in $D$\\
    \hline
    $\mathcal{A}$ & Set of attributes in $D$\\
    \hline
    $Dom(\cdot)$ & Domain of a set of attributes\\
    \hline
    $score(t)$ & Score of the tuple $t$ computed using a monotonic function $F(\cdot)$\\
    \hline
    $t \succ t^\prime$ & $t$ dominates $t^\prime$\\
    \hline
    $\mathcal{L}$ & Set of $m$ sorted lists\\
    \hline
    $\mathcal{Q}$ & Subspace skyline query\\
    \hline
    $m^\prime$ & Number of attributes in $\mathcal{Q}$\\
    \hline
    $D_{\mathcal{Q}}$ & Projection of $D$ in query space $\mathcal{Q}$\\
    \hline
    $\mathcal{S}_\mathcal{Q}$ & Set of skyline tuples in $D_{\mathcal{Q}}$\\
    \hline
    $t_{\mathcal{Q}}$ & Projection of tuple $t$ in $\mathcal{Q}$\\
    \hline
    $t_{\mathcal{Q}} \succ_{\mathcal{Q}} t^\prime_{\mathcal{Q}}$ & $t_{\mathcal{Q}}$ dominates $t^\prime_{\mathcal{Q}}$ on query space $\mathcal{Q}$ \\
    \hline
    $\mathcal{L_Q}$ & Set of sorted lists corresponds to attributes in $\mathcal{Q}$\\
    \hline
    $cv_{ij}$ & Attribute value returned by $i$-th sorted access on list $L_j$\\
    \hline
    $T$ & Tree for storing the candidate skyline tuples\\
    \hline
    $p_i$ & the probability that the binary attribute $A_i$ is $1$\\
    \hline
\end{tabular}
\end{tiny}
\end{table}

\section{Skyline Computation Over Categorical Attributes}\label{sec:3}
Without loss of generality, for ease of explanation, we consider a relation with Boolean attributes, i.e., categorical attributes with domain size 2. We shall discuss the extensions of the algorithms for categorical attributes with larger domains later in this section.

Throughout this section, we consider the case in which precomputed indices are not available. First, we exploit the categorical characteristics of attributes by designing a tree data structure that can perform efficient {\em dominance} operations. Specifically, given a new tuple $t$, the tree supports three primitive operations -- i) INSERT($t$): inserts a new tuple $t$ to the tree, ii) IS-DOMINATED($t$): checks if tuple $t$ is dominated by any tuple in the tree, and iii) PRUNE-DOMINATED-TUPLES($t$): deletes the tuples dominated by $t$ from the tree. In Appendix \ref{ap:tree-optimizations}, we further improve the performance of these basic operations by proposing several optimization techniques. Finally, we propose two algorithms ST-S (Skyline using Tree Sorting-based) and ST-P (Skyline using Tree Partition-based) that incorporate the tree structure to state-of-art sorting- and partition-based algorithms.

\subsection{Organizing Tuples Tree}\label{subsec:tree}
\vspace{1mm}
\noindent{\bf Tree structure:} We use a binary tree to store tuples in the candidate skyline set. Consider an ordering of all attributes in $\mathcal{Q} \subseteq \mathcal{A}$, e.g., $[A_1, A_2, \ldots, A_{m'}]$.
In addition to tuple attributes, we enhance each tuple with a score, assessed using a function $F(\cdot)$. This score assists in improving performance during identification of the dominated tuples or while conducting the dominance check. The proposed algorithm is agnostic to the choice of $F(\cdot)$; the only requirement is that the function does not assign a higher score to a dominated tuple compared to its dominator.
The structure of the tree for Example~\ref{exmp:ST} is depicted in Figure~\ref{fig:tree}. The tree has a total of 5 ($=m'+ 1$) levels, where the $i$'th level ($1 \leq i \leq m'$) represents attribute $A_i$. The left (resp. right) edge of each internal node represents value 0 (resp. 1). Each path from the root to a leaf represents a specific assignment of attribute values. The leaf nodes of the tree 
store two pieces of information: i) $score$: the score of the tuple mapped to that node, and ii) \textit{tupleID List}: list of ids of the tuples mapped to that node. Note that all the tuples that are mapped to the same leaf node in the tree have the same attribute value assignment, i.e. have the same score.
Moreover, the attribute values of a tuple $t$ can be identified by inspecting the path from the root to a leaf node containing $t$. Thus, there is no requirement to store the attribute values of the tuples in the leaf nodes.
Only the leaf nodes that correspond to an actual tuple are present in the tree. 

\begin{exmp}\label{exmp:ST} 
As a running example through out this section, consider the relation $D$ with $n=5$ non-dominated tuples where its projection on $\mathcal{Q}=\{A_1,A_2,A_3,A_4\}$ is depicted in Table~\ref{tab:skylineTreeRunningExample}. 
The last column of the table presents the score of each tuple, utilizing the function $F(\cdot)$  provided in Equation~\ref{eq:score}.
\begin{align}\label{eq:score}
F(t_{\mathcal{Q}}) = \sum_{A_i \in \mathcal{Q}} 2^{i-1} \cdot t[A_i]
\end{align}
\end{exmp}

\begin{table}[!t]
\centering
\caption{Example~\ref{exmp:ST} relation}\label{tab:skylineTreeRunningExample}
\begin{tiny}
\begin{tabular}{cccccc}
    \hline 
    $tupleID$ & $A_1$ & $A_2$ & $A_3$ & $A_4$ & $Score$\\
    \hline 
    $t_1$ & 1 & 1 & 0 & 0 & 12\\
    \hline
    $t_2$ & 0 & 0 & 1 & 1 & 3\\
    \hline
    $t_3$ & 0 & 1 & 1 & 0 & 6\\
    \hline
    $t_4$ & 1 & 0 & 0 & 1 & 9\\
    \hline
    $t_5$ & 1 & 0 & 1 & 0 & 10\\
    \hline
\end{tabular}
\end{tiny}
\end{table}

\begin{figure*}[!t]
\begin{minipage}[t]{0.23\linewidth}
\centering
    \includegraphics[scale=0.80]{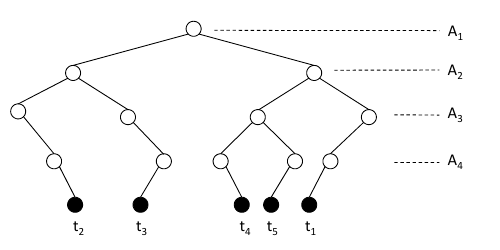}
    \caption{Tree structure for relation in Example~\ref{exmp:ST}}
    \label{fig:tree}
\end{minipage}
\hspace{1mm}
\begin{minipage}[t]{0.23\linewidth}
\centering
    \includegraphics[scale=0.80]{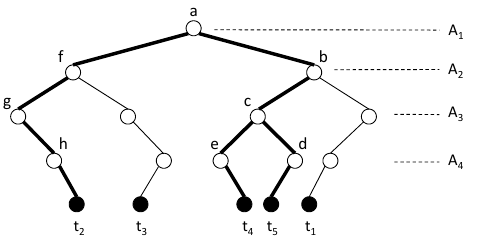}
    \caption{Prune dominated tuples}
    \label{fig:treePruneDominatedTuples}
\end{minipage}
\hspace{1mm}
\begin{minipage}[t]{0.23\linewidth}
\centering
    \includegraphics[scale=0.80]{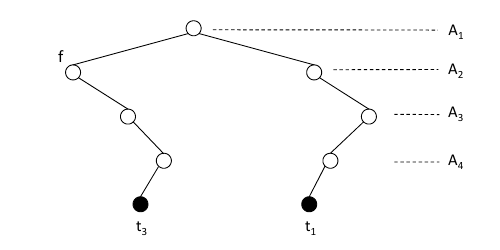}
    \caption{Tree after removing dominated tuples}
    \label{fig:treePruneDominatedTuplesAfter}
\end{minipage}
\hspace{1mm}
\begin{minipage}[t]{0.23\linewidth}
\centering
    \includegraphics[scale=0.80]{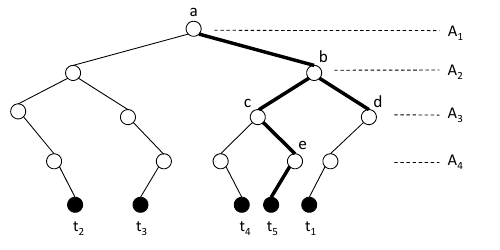}
    \caption{Check if tuple $t$ is dominated}
    \label{fig:treeCheckIfDominated}
\end{minipage}
\end{figure*}

\vspace{1mm}
\noindent{\bf INSERT($t$):} In order to insert a tuple $t$ into the tree, we start from the root. At level $i$ $(1 \leq i \leq m')$, we check the corresponding attribute value, $t[A_i]$. If $t[A_i] = 0$ (resp. $t[A_i] = 1$) and the left (resp. right) child of current node already exists in the tree, we simply follow the left (resp. right) child. Otherwise, we first have to create a new tree node as left (resp. right) child before traversing it. After reaching the leaf node at level $m'+1$, the $tupleID$ of $t$ is appended to \textit{tupleID List} and the $score$ value is assigned to newly constructed leaf.

\begin{algorithm}[!htb]
\caption{{\bf INSERT}}
\begin{algorithmic}[1]
\label{alg:insertTuple}
\STATE {\bf Input:} Tuple $t$, Node $n$, Level $l$, Query $\mathcal{Q}$;
\STATE {\bf if} $l == |\mathcal{Q}| + 1$:
    \STATE \hindent {\bf if} $n.score$ is None: $n.score = score(t)$
    \STATE \hindent Append $t[tupleID]$ to $n.tupleIDList$
\STATE {\bf else}:
    \STATE \hindent {\bf if} $t[A_l]==0$:
        \STATE \hindent[2] {\bf if} $n.left$ is $None$:
           \STATE \hindent[3] temp = {\it New} Node();
           \STATE \hindent[3] $t.left$ = temp;
        \STATE \hindent[2] INSERT($t$, $n.left$, $l+1$)
    \STATE \hindent {\bf if} $t[A_l]==1$:
        \STATE \hindent[2] {\bf if} $n.right$ is $None$:
           \STATE \hindent[3] temp = {\it New} Node();
           \STATE \hindent[3] $t.right$ = temp;
        \STATE \hindent[2] INSERT($t$, $n.right$, $l+1$)
\end{algorithmic}
\end{algorithm}

\vspace{1mm}
\noindent{\bf PRUNE-DOMINATED-TUPLES($t$):} The pruning algorithm to delete from the tree, tuples dominated by $t$, is recursively developed as follows: We start from the root node of the tree. If $t[A_1] = 1$, we search both the left and right subtree. Otherwise, only the left child is selected. This is because if $t[A_1] = 1$, a tuple $t'$ dominated by $t$ can assume value  0 or 1 on attribute $A_1$. On the other hand, $t$ cannot dominate a tuple $t'$ if $t[A_1] = 0$ and $t'[A_1] = 1$. We follow the same approach at each internal node visited by the algorithm - at level $i$ $(1 \leq i \leq m)$, value of $t[A_i]$ is used to select the appropriate subtree. After reaching a leaf node, we compare $score(t_{\mathcal{Q}})$ with the $score$ value of leaf node. If both values are equal, no action is required, since, all the tuples mapped into the current leaf node have the same attribute value as $t_{\mathcal{Q}}$. Else, the leaf node is deleted from the tree. Upon return from the recursion, we check if both the left and right child of the current (internal) node are empty. In that case, the current node is also deleted from the tree.

Figure~\ref{fig:treePruneDominatedTuples} demonstrates the pruning algorithm for $t = \langle 1,0,1,1 \rangle$. Tuples in the tree that are dominated by $t$ are: $t_2$, $t_4$, and $t_5$. The bold edges represent paths followed by the pruning algorithm. Both the left and right children of node $a$ are visited since $t[A_1] = 1$, whereas, at nodes $f$ and $b$ only the left subtree is selected for searching. The final structure of the tree after deleting the dominated tuples is shown in Figure~\ref{fig:treePruneDominatedTuplesAfter}.

\begin{algorithm}[htb]
\caption{{\bf PRUNE-DOMINATED-TUPLES}}
\begin{algorithmic}[1]
\label{alg:pruneDominatedTuples}
\STATE {\bf Input:} Tuple $t$, Node $n$, Level $l$, Score $s$, Query $\mathcal{Q}$;

\STATE {\bf if} $n$ is $None$ or $n.minScore > s$ {\bf return}

\STATE {\bf if} $l == |\mathcal{Q}| + 1$ and $score(t_{\mathcal{Q}}) \neq n.score$:
    \STATE \hindent Delete $n$ from tree
    \STATE \hindent {\bf return}

\STATE {\bf if} $t[A_l] == 1$:
    \STATE \hindent PRUNE-DOMINATED-TUPLES($t$, $n.right$, $l+1$, $s$)
    \STATE \hindent $s' = s - weight(A_i)$
    \STATE \hindent PRUNE-DOMINATED-TUPLES($t$, $n.left$, $l+1$, $s'$)
\STATE {\bf else}:
    \STATE \hindent PRUNE-DOMINATED-TUPLES($t$, $n.left$, $l+1$, $s$)

\STATE {\bf if} Both $left$ and $right$ children of $n$ is $None$
    \STATE \hindent Delete $n$ from tree
\end{algorithmic}
\end{algorithm}

\vspace{1mm}
\noindent{\bf IS-DOMINATED($t$):} The algorithm starts traversing the tree from the root. For each node visited by the algorithm at level $i$ $(1 \leq i \leq m)$, we check the corresponding attribute value $t[A_i]$. If $t[A_i] = 0$, we search both the left and right subtree; otherwise, we only need to search in the right subtree. This is because when $t[A_i] = 0$, all the tuples dominating $t$ can be either 0 or 1 on attribute $A_i$. If we reach a leaf node that has an attribute value assignment which is different than that of $t$ (i.e., $score \neq score(t)$), $t$ is dominated.  Note that, when $t[A_i] = 0$ both the left and right subtree of the current node can have tuples dominating $t$, while the cost of identifying a dominating tuple (i.e., the number of nodes visited) may vary depending on whether the left or right subtree is visited first. For simplicity, we always search in the right subtree first. If there exists a tuple in the subtree of a node that dominates tuple $t$, we do not need to search in the left subtree anymore. 

Figure~\ref{fig:treeCheckIfDominated} presents the nodes visited by the algorithm in order to check if the new tuple $t = \langle 0,0,1,0 \rangle$ is dominated. We start from the root node $a$ and check the value of $t$ in attribute $A_1$. Since $t[A_1] = 0$, we first search in the right subtree of $a$. After reaching to node $d$, the algorithm back-tracks to $b$ (parent of $d$). This is because $t[A_3] = 1$ and $d$ has no actual tuple mapped under it's right child. Since $t[A_2] = 0$ and we could not identify any dominating tuple in the right subtree of $b$, the algorithm starts searching in the left subtree and moves to node $c$. At node $c$, only the right child is selected, since $t[A_3] = 1$. Applying the same approach at node $f$, we reach the leaf node $e$ that contains the tupleID $t_5$. Since the value of the $score$ variable at leaf node $e$ is different from $score(t)$, we conclude that tuples mapped into $e$ (i.e., $t_5$) dominate $t$.

Please refer to Appendix~\ref{ap:tree-optimizations} for further optimizations on the tree data structure.

\begin{algorithm}[htb]
\caption{{\bf IS-DOMINATED}}
\begin{algorithmic}[1]
\label{alg:isDominated}
\STATE {\bf Input:} Tuple $t$, Node $n$, Level $l$, Score $s$, Query $\mathcal{Q}$; \qquad {\bf Output:} True if $t$ is dominated else False.
\STATE {\bf if} $n$ is $None$ or $s > n.maxScore$: {\bf return}

\STATE {\bf if} $l == |\mathcal{Q}|$ and $score(t_{\mathcal{Q}}) \neq n.score$: {\bf return} True
\STATE {\bf if} $l == |\mathcal{Q}|$ and $score(t_{\mathcal{Q}}) = n.score$: {\bf return} False

\STATE {\bf if} $t[A_l] == 0$:
    \STATE \hindent $s' = s + weight(A_i)$
    \STATE \hindent $dominated$ = IS-DOMINATED($t$, $n.right$, $l+1$, $s'$)
    \STATE \hindent {\bf if} $dominated$ == True: {\bf return} True
    \STATE \hindent {\bf return} IS-DOMINATED($t$, $n.left$, $l+1$, $s$)
\STATE {\bf else}:
    \STATE \hindent {\bf return} IS-DOMINATED($t$, $n.right$, $l+1$, $s$)
\end{algorithmic}
\end{algorithm}

\subsection{Skyline using Tree}\label{sec:ST}

Existing works on skyline computation mainly focus on two optimization criteria: reducing the number of dominance checks (CPU cost), limiting communication cost with the backend database (I/O cost). Sorting-based algorithms reduce the number of dominance check by ensuring that only the skyline tuples are inserted in the candidate skyline list. Whereas, partition-based algorithms achieve this by skipping dominance tests among tuples inside incomparable regions generated from the partition. However, given a list of tuples $\mathcal{T}$ and a new tuple $t$, in order to discard tuples from $\mathcal{T}$ that are dominated by $t$, both the sorting- and partition-based algorithms need to compare $t$ against all the tuples in $\mathcal{T}$. This is also the case when we need to check whether $t$ is dominated by $T$. The tree structure defined in \S\ref{subsec:tree} allows us to perform these operations effectively for categorical attributes. Since the performance gain achieved by the tree structure is independent of the optimization approaches of previous algorithms, it is possible to combine the tree structure with existing skyline algorithms. We now present two algorithms ST-S (Skyline using Tree Sorting-based) and ST-P (Skyline using Tree Partition-based) that incorporates the tree structure into existing algorithm.

\vspace{1mm}
\noindent{\bf ST-S:} ST-S combines the tree structure with a sorting-based algorithm. Specifically, we have selected the SaLSa~\cite{bartolini2008efficient} algorithms that exhibits better performance compared to other sorting-based algorithms. The final algorithm is presented in Algorithm~\ref{alg:st-s}. The tuples are first sorted according to ``maximum coordinate'', maxC, criterion\footnote{Assuming larger values are preferred for each attribute.}. Specifically, Given a skyline query $\mathcal{Q}$, $maxC(t_{\mathcal{Q}}) = (max_{A\in \mathcal{Q}}\{t[A]\}, sum(t_{\mathcal{Q}}))$, where $sum(t_{\mathcal{Q}}) = \sum_{A\in \mathcal{Q}} t[A]$. A tree structure $T$ is used to store the skyline tuples. Note that the monotonic property of the scoring function $maxC(\cdot)$ ensures that all the tuples inserted in $T$ are skyline tuples. The algorithm then iterates over the sorted list one by one, and for each new tuple $t$, if $t$ is not dominated by any tuple in tree $T$, it is inserted in the tree (lines 7-8). For each new skyline tuple, the ``stop point'' $t_{stop}$ is updated if required (line 10-12). The algorithm stops if all the tuples are accessed or $t_{stop}$ dominates the remaining tuple. Detailed description of the ``stop point'' can be found in the original SaLSa paper~\cite{bartolini2008efficient}.

\begin{algorithm}[htb]
\caption{{\bf ST-S}}
\begin{algorithmic}[1]
\label{alg:st-s}
\STATE {\bf Input:} Tuple list $\mathcal{T}$, Query $\mathcal{Q}$ and Tree $T$; \\ {\bf Output:} $\mathcal{S}_\mathcal{Q}$
\STATE Sort tuples in $D$ using a monotonic function $maxC(\cdot)$
\STATE {\bf if} $T \text{ is } None$: $T \leftarrow$ {\it New} Tree()
\STATE $t_{stop} \leftarrow$ undefined
\STATE {\bf for} each tuple $t \in D$
    \STATE \hindent {\bf if} $t_{stop}^+ \geq maxC(t_{\mathcal{Q}})$ and $t_{stop} \neq t$: {\bf return}
    \STATE \hindent {\bf if not} IS-DOMINATED($t_\mathcal{Q}$, $T.rootNode$, 1, $score(t)$)
        \STATE \hindent[2] INSERT($t_\mathcal{Q}$, $T.rootNode$, 1)
        \STATE \hindent[2] Output $t_\mathcal{Q}$ as skyline tuple.
        \STATE \hindent[2] $t^+ \leftarrow  min_{A \in \mathcal{Q}}\{t[A]\}$
        \STATE \hindent[2] {\bf if} $t^+ > t_{stop}^+$: $t_{stop} \leftarrow t_{\mathcal{Q}}$
\end{algorithmic}
\end{algorithm}

\vspace{1mm}
\noindent{\bf ST-P:} We have selected the state-of-art partition-based algorithm BSkyTree~\cite{lee2014scalable} for designing ST-P. The final algorithm is presented in Algorithm~\ref{alg:st-p}. Given a tuple list $\mathcal{T}$, the SELECT-PIVOT-POINT method returns a pivot tuple $p^V$ such that it belongs to the skyline of $\mathcal{Q}$ (i.e., $\mathcal{S_{\mathcal{Q}}}$). Moreover, $p^V$ partitions the tuples in $\mathcal{T}$ in a way such that the number of dominance test is minimized (details in~\cite{lee2014scalable}). Tuples in $\mathcal{T}$ are then split into $2^{|\mathcal{Q}|}$ lists, each corresponding to one of the $2^{|\mathcal{Q}|}$ regions generated by $p^V$ (lines 7-9). Tuples in $\mathcal{L}[0]$ are dominated by $p^V$, hence can be pruned safely. For each pair of lists $\mathcal{L}[i]$ and $\mathcal{L}[j]$ ($max \geq j> i \geq 1$), if $\mathcal{L}[j]$ partially dominates $\mathcal{L}[i]$, tuples in $\mathcal{L}[i]$ that are dominated by any tuple in $\mathcal{L}[j]$ are eliminated. Finally, skylines in $\mathcal{L}[i]$ are then discovered in recursive manner (lines 10-15).

\begin{algorithm}[htb]
\caption{{\bf ST-P}}
\begin{algorithmic}[1]
\label{alg:st-p}
\STATE {\bf Input:} Tuple list $\mathcal{T}$ and query $\mathcal{Q}$; \\ {\bf Output:} $\mathcal{S}_\mathcal{Q}$
\STATE {\bf if} $|\mathcal{T}| \leq 1$: {\bf return} $\mathcal{T}$
\STATE $max \leftarrow 2^{|\mathcal{Q}|} - 2$ //\textit{Size of the lattice}
\STATE $\mathcal{L}[1, max] \leftarrow \{\}$ 
\STATE $p^V \leftarrow$ SELECT-PIVOT-POINT($\mathcal{T}$)
\STATE $\mathcal{S_\mathcal{Q}} \leftarrow \mathcal{S_\mathcal{Q}} \cup p^V$ //\textit{$p^V$ is a skyline tuple}
\STATE {\bf for} each tuple $t \in \mathcal{T}$
    \STATE \hindent $B^i \leftarrow$ $|\mathcal{Q}|$-bit binary vector corresponds to
    $t$ wrt $p^V$
    \STATE \hindent {\bf if} $i \neq 0$: $\mathcal{L}[i] \leftarrow \mathcal{L}[i] \cup t$
\STATE {\bf for} $i \leftarrow \text{ max to } 1$
    \STATE \hindent $T \leftarrow$ {\it New} Tree()
    \STATE \hindent Insert tuples in $\mathcal{L}[i]$ in $T$
    \STATE \hindent {\bf for} $\forall j \in [max, i)$ : $B^j \succeq B^i$
        \STATE \hindent[2] {\bf for} $\forall t \in \mathcal{L}[j]$: PRUNE-DOMINATED-TUPLES($t_{\mathcal{Q}}$, $T.rootNode$, 1, $score(t_{\mathcal{Q}})$)
    \STATE \hindent $\mathcal{S_\mathcal{Q}} \leftarrow \mathcal{S_\mathcal{Q}} \cup $ ST-P(tuples in $T$)
\STATE {\bf return} $\mathcal{S_{\mathcal{Q}}}$
\end{algorithmic}
\end{algorithm}

\vspace{1mm}
\noindent{\bf Performance Analysis:} We now provide a theoretical analysis of the performance of primitive operations utilized by ST-S and ST-P. To make the theoretical analysis tractable, we assume that the
underlying data is i.i.d., where $p_i$ is the probability of having value 1 on attribute $A_i$.

The cost of INSERT-TUPLE($t_\mathcal{Q}$) operation is $O(m')$, since to insert a new tuple in the tree one only needs to follow a single path from the root to leaf. For IS-DOMINATED($t_\mathcal{Q}$) and PRUNE-DOMINATED-TUPLES($t_\mathcal{Q}$), we utilize the number of nodes visited in the tree as the performance measure of these operations.

Consider a tree $T$ with $s$ tuples;  Let $Cost(l, s)$ be the expected number of nodes visited by the primitive operations.

\begin{theorem}\label{thm:expectedCostSTISDominated}
Considering a relation with $n$ binary attributes where $p_i$ is the probability that a tuple has value 1 on attribute $A_i$, the expected cost of IS-DOMINATED($t_\mathcal{Q}$) operation on a tree $T$, containing $s$ tuples is:
\begin{small}
\begin{align}\label{eq:expectedCostSTISDominated}
    \nonumber
    C(m', s) &= 1 \\
    \nonumber
    C(l, 0) &= 1 \\
    C(l, s) &= 1 + \sum_{i=0}^s {s \choose i} (1-p_l)^i p_l^{s-i} C'(l, i, s-i)
\end{align}
\end{small}
\hspace{-1mm}where $S(l, s-i) = 1 - (1 - \prod\nolimits_{i=1}^{|\mathcal{A}_{ones(t[l+1:m'])}|}p_i)^{s-i}$ and\footnote{$\mathcal{A}_{ones(t[l+1:m'])} = \{A_i | t[A_i] = 1, l+1 \leq i \leq m'\}$ is the set of remaining attributes of $t$ that has value equals 1.} $C'(l, i, s-i) = C(l+1, s-i) + (1-p_l)(1-S(l, s-i))C(l+1, i)$
\end{theorem}
Please refer to Appendix~\ref{sec:appendixProof} for the proof.

\begin{theorem}\label{thm:expectedCostSTPruneDominatedTuples}
Given a boolean relation $D$ with $n$ tuple and the probability of having value 1 on attribute $A_i$ being $p_i$, the expected cost of PRUNE-DOMINATED-TUPLES($t_\mathcal{Q}$) operation on a tree $T$, containing $s$ tuples is
\begin{small}
\begin{align} \label{eq:expectedCostSTPruneDominatedTuples}
    \nonumber
    C(m', s) &= 1 \\
    \nonumber
    C(l, 0) &= 1 \\
    C(l, s) &= 1 + \sum_{i=0}^s {s \choose i} (1-p_l)^i p_l^{s-i} (C(l+1, i) + p_lC(l+1, s-i))
\end{align}
\end{small}
\end{theorem}
The proof is available in Appendix~\ref{sec:appendixProof}

Figure~\ref{fig:expectedCostST}
uses Equations~\ref{eq:expectedCostSTISDominated} and~\ref{eq:expectedCostSTPruneDominatedTuples} to provide an expected cost for the IS-DOMINATE and PRUNE-DOMINATED-TUPLES operations, for varying numbers of tuples in $T$ ($s$) where $m'=20$.
We compare its performance with the appraoch, where candidate skyline tuples are organized in a list.
Suppose there are $s$ tuples in the list; the best case for the domination test occurs when the first tuple in the list dominates the input tuple ($O(1\times m')$), while in the worst case, none or only the very last tuple dominates it ($O(s\times m')$)~\cite{borzsony2001skyline}. Thus, on average the dominance test iterates over half of its candidate list (i.e., $\dfrac{s}{2}\times m'$ comparisons).
On the other hand, in order to prune tuples in the list that are dominated by $t_\mathcal{Q}$, existing algorithms need to compare $t_\mathcal{Q}$ with all the entries in the list. Hence, expected cost of PRUNE-DOMINATED-TUPLES is $s \times m'$. From the figure, we can see that the expected number of comparisons required by the two primitive operations are significantly less when instead of a list, tuples are organized in a tree. Moreover, as $p_i$ increases, the cost of the primitive operations decreases. This is because, when the value of $p_i$ is large, the probability of following left edge (edges corresponds value $0$) of a tree node decreases. 



The above simulations show that the tree structure can reduce the cost of dominance test effectively thus improving the overall performance of ST algorithms. Although the analysis has been carried out for i.i.d. data, our experimental results in \S\ref{sec:experiments} show similar behavior for other types of datasets.

\begin{figure}
\begin{subfigure}{.49\linewidth}
  \centering
  \includegraphics[scale=.45]{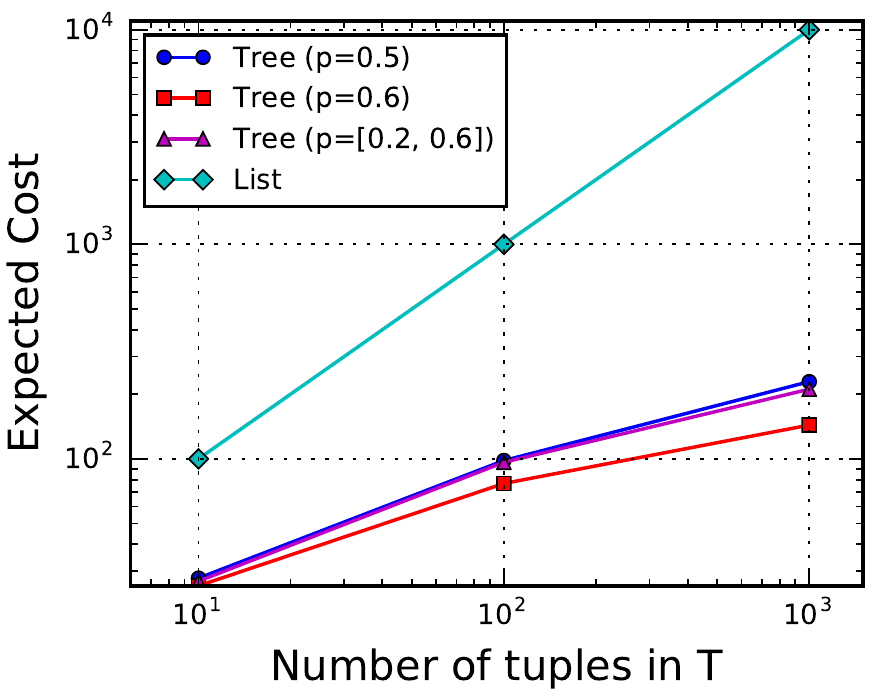}
  \caption{\begin{tiny}IS-DOMINATED\end{tiny}}
  \label{fig:expectedCostIsDominated}
\end{subfigure}
\begin{subfigure}{.49\linewidth}
  \centering
  \includegraphics[scale=.45]{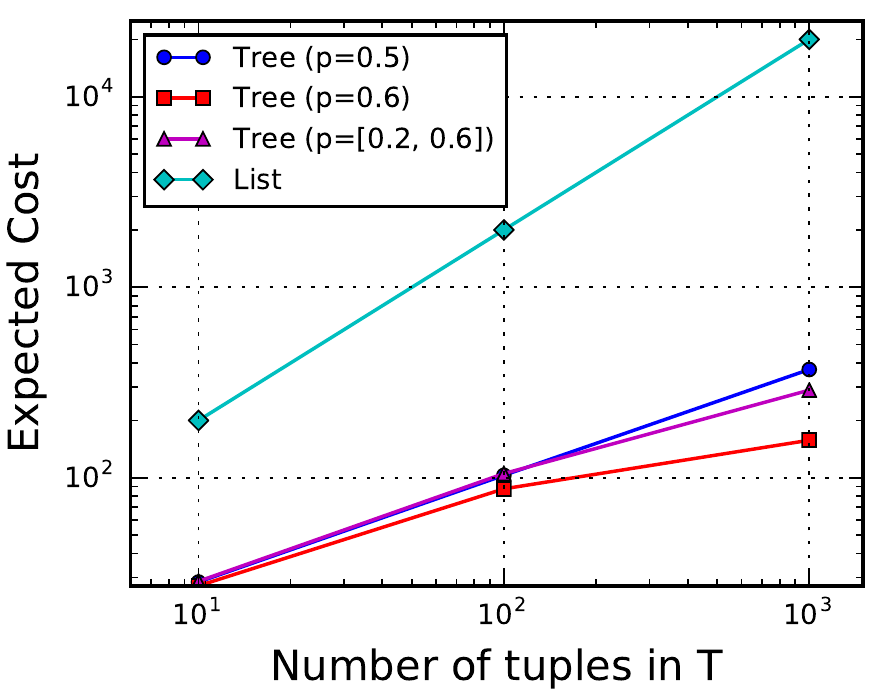}
  \caption{\begin{tiny}PRUNE-DOMINATED-TUPLES \end{tiny}}
  \label{fig:expectedPrunedDominated}
\end{subfigure}
\caption{Expected cost of IS-DOMINATED and PRUNE-DOMINATED-TUPLES operations as a function of $s$}
\label{fig:expectedCostST}
\end{figure}

\subsection{Extension for Categorical Attributes}\label{ap:STCategorical}
We now discuss how to modify ST algorithm for relations having categorical attributes. We need to make the following two changes:

\begin{itemize}
    \item The tree structure designed in \S\ref{subsec:tree} needs to be modified for categorical attribute.
    \item We also need to change the tree traversal algorithms used in each of the three primitive operations.
\end{itemize}

\noindent{\bf Tree structure:} The tree structure will not be binary anymore. In order to incorporate categorical attributes, each node $u$ at level $l$ ($1 \leq l \leq m$) of the tree now should have $|Dom(A_l)|$ children, one for each attribute value $v \in Dom(A_l)$. We shall index the edges from left to right, where the left most edge corresponds to the lowest attribute value and the attribute value corresponding to each edge increases as we move from left most edge to right most edge.

\vspace{1mm}
\noindent{\bf INSERT($t$):} After reaching a node $u$ at level $l$, select the $t[A_l]$-th child of $u$ for moving to the next level of the tree.

\vspace{1mm}
\noindent{\bf IS-DOMINATED($t$):} We need to follow all the edges that has index value grater or equal to $t[A_l]$.

\vspace{1mm}
\noindent{\bf PRUNE-DOMINATED-TUPLES($t$):} Search in all the subtrees reachable by following edges with index value less than or equal to $t[A_l]$.

\section{Subspace Skyline using Sorted \\ Lists} \label{sec:subsky}

In this section, we consider the availability of sorted lists $L_1, L_2, \ldots L_m$, as per \S\ref{sec:preliminaries} and utilize them to design efficient algorithms for subspace skyline discovery.
We first briefly discuss a baseline approach that is an extension of LS~\cite{morse2007efficient}.
Then in \S\ref{sec:topdown}, we overcome the barriers of the baseline approach proposing an algorithm named {\bf TOP-DOWN}. The algorithm applies a top-down on-the-fly parsing of the subspace lattice and prunes the dominated branches.
However, the expected cost of TOP-DOWN {\em exponentially} depends on the value of $m$  (Appendix~\ref{ap:top-down}).
We then propose {\bf TA-SKY} (Threshold Algorithm for Skyline) in \S\ref{sec:TASky} that does not have such a dependency. In addition to the sorted lists, TA-SKY also utilizes the ST algorithm proposed in \S\ref{sec:3} for computing skylines.

\begin{table}[!t]
\centering
\caption{Example: Input Table}\label{tab:runningExampleSubspaceSkyline}
\begin{tiny}
\begin{tabular}{cccccc}
    \hline 
     & $A_1$ & $A_2$ & $A_3$ & $A_4$ & $A_5$ \\
    \hline 
    $t_1$ & 0 & 1 & 0 & 1 & 1\\
    \hline
    $t_2$ & 0 & 0 & 1 & 1 & 0\\
    \hline
    $t_3$ & 0 & 0 & 1 & 0 & 1\\
    \hline
    $t_4$ & 0 & 0 & 0 & 1 & 1\\
    \hline
    $t_5$ & 1 & 0 & 1 & 1 & 1\\
    \hline
    $t_6$ & 1 & 1 & 1 & 0 & 0\\
    \hline
\end{tabular}
\end{tiny}
\end{table}

\begin{figure}[!ht]
  \begin{minipage}[t]{0.49\linewidth}
    \centering
    \includegraphics[scale=1.2]{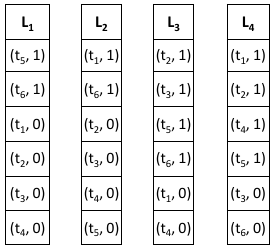}
    \caption{Example: Sorted Lists, Organization 1} 
    \label{fig:sortedLists}
  \end{minipage}
  \hspace{1mm}
  \begin{minipage}[t]{0.49\linewidth}
    \centering
    \includegraphics[scale=1.2]{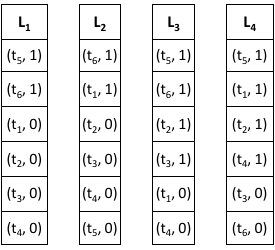}
    \caption{Example: Sorted Lists, Organization 2}
    \label{fig:sortedListsOptimized}
  \end{minipage}
\end{figure}

\begin{exmp}\label{exmp:subspaceSkyline}
Let $\mathcal{Q} \subseteq \mathcal{A}$ denotes the set of attributes in a subspace skyline query and $D_{\mathcal{Q}}$ be the projection of $D$ in $\mathcal{Q}$. We denote the set of sorted lists corresponding to a query (one for each attribute involved in the query) as $\mathcal{L_Q}$, $\mathcal{L_Q} = \{ L_i | A_i \in \mathcal{Q} \}$. Also, let $m' \leq m$ be $|\mathcal{Q}|$. Our running example uses the relation shown in Table~\ref{tab:runningExampleSubspaceSkyline} through out this section. There are a
total of  $n=6$ tuples, each having $m=5$ attributes. Consider a subspace skyline query $\mathcal{Q} = \{A_1, A_2, A_3, A_4\}$, thus, $m' = 4$. Figure~\ref{fig:sortedLists} shows the corresponding sorted lists $\mathcal{L_Q} = \{L_1, L_2, L_3, L_4 \}$.
\end{exmp}

\vspace{1mm}
\noindent{\bf BASELINE:} We use sorted lists in $\mathcal{L_Q}$ to construct the projection of each tuple $t \in D$ in the query space. For this, we shall perform $n$ sequential accesses on sorted list $L_1 \in \mathcal{L_Q}$. For each $(tupleID, value)$ pair returned by sequential access, we create a new tuple $t_{new}$. $t_{new}$ has $tupleID$ as its id and $t_{new}[A_1] = value$. The remaining attribute values of $t_{new}$ are set by performing random access on sorted list $L_j$ ($\forall j \in [2,m']$). After computing the projections of all tuples in query space, we create a lattice over $\mathcal{Q}$ and 
run the LS algorithm to discover the subspace skyline.

\vspace{3mm}
\noindent We identify the following problems with BASELINE:
\begin{itemize}
\itemsep0em
\item It makes two passes over all the tuples in the relation.
\item It requires the construction of the complete lattice of size $|Dom(\mathcal{Q})|$. For example, when $Dom(A_i) = 4$ and $m'=15$, the lattice has more than {\em one billion} nodes; yet the algorithm needs to map the tuples into the lattice.
\end{itemize}

One observation is that for relations with categorical attributes, especially when $m'$ is relatively small, skyline tuples are more likely to be discovered at the upper levels of the lattice. This motivated us to seek alternate approaches.
Unlike BASELINE, TOP-DOWN and the TA-SKY algorithm are designed in a way that they are capable of answering subspace skyline queries by traversing a small portion of the lattice, and more importantly {\em without the need to access the entire relation}.

\subsection{TOP-DOWN}\label{sec:topdown}

\noindent{\bf Key Idea:} Given a subspace skyline query $\mathcal{Q}$, we create a lattice capturing the dominance relationships among the tuples in $D_{\mathcal{Q}}$. Each node in the lattice represents a specific attribute value combination in query space, hence, corresponds to a potential tuple in $D_{\mathcal{Q}}$. For a given lattice node $u$, if there exist tuples in $D_{\mathcal{Q}}$ with attribute value combination same as $u$, then all tuples in $D_{\mathcal{Q}}$ corresponding to nodes dominated by $u$ in the lattice are also dominated. TOP-DOWN utilizes this observation to compute skylines for a given subspace skyline query. Instead of iterating over the tuples, TOP-DOWN traverses the lattice nodes from top to bottom; it utilizes sorted lists $\mathcal{L_Q}$ to search for tuples with specific attribute value combinations. When $|\mathcal{Q}|$ is relatively small, it is likely one will discover all the skyline tuples just by checking few attribute value combinations, without considering the rest of the lattice. However, the expected cost of TOP-DOWN increases exponentially as we increase the query  length.
Please refer to Appendix~\ref{ap:top-down} for the details and the limitations of TOP-DOWN.

\subsection{TA-SKY}\label{sec:TASky}
We now propose our second algorithm, Threshold Algorithm for Skyline (TA-SKY) in order to answer subspace skyline queries. Unlike TOP-DOWN that exponentially depends on $m$, as we shall show in \S\ref{sec:TASKY-performance}, TA-SKY has a worst case time complexity of $O(m'n^2)$; in addition, we shall also study the expected cost of TA-SKY.
The main innovation in TA-SKY is that it follows the style of the well-known Threshold Algorithm (TA)~\cite{fagin2003optimal} for Top-$k$ query processing, except that it is used for solving a skyline problem rather than a Top-$k$ problem. 

TA-SKY iterates over the sorted lists $\mathcal{L_Q}$ until a stopping condition is satisfied. At each iteration, we perform $m'$ parallel sorted access, one for each sorted list in $\mathcal{L_Q}$. Let $cv_{ij}$ denote the current value returned from sorted access on list $L_j \in \mathcal{L_Q}$ $(1 \leq j \leq m')$ at iteration $i$. Consider $\tau_i$ be the set of values returned at iteration $i$, $\tau_i = \{cv_{i1}, cv_{i1}, \ldots, cv_{im'}\}$. We create a synthetic tuple $t_{syn}$ as the \textit{threshold value} to establish a stopping condition for TA-SKY. The attribute values of synthetic tuple $t_{syn}$ are set according to the current values returned by each sorted list. Specifically, at iteration $i$, $t_{syn}[A_j] = cv_{ij}, \forall j \in [1, m']$. In other words, $t_{syn}$ corresponds to a potential tuple with the highest possible attribute values that has not
been seen by TA-SKY yet. 

In addition, TA-SKY also maintains a candidate skyline set. The candidate skyline set materializes the skylines among the tuples seen till the last stopping condition check. We use the tree structure described in \S\ref{sec:ST} to organize the candidate skyline set. Note that instead of checking the stopping condition at each iteration, TA-SKY considers the stopping condition at iteration $i$ only when $\tau_i \neq \tau_{i-1}$ $(2 \leq i \leq n)$.  $\tau_i \neq \tau_{i-1}$ if and only if $cv_{(i-1)j} \neq cv_{ij}$ $(1 \leq j \leq m')$ for at least one of the $m'$ sequential accesses. This is because the stopping condition does not change among iterations that have the same $\tau$ value. Let us assume the value of $\tau$ changes at the current iteration $i$ and the stopping condition was last checked at iteration $i'$ ($i' < i)$. Let $\mathcal{T}$ be the set of tuples that are returned in, at least one of the sequential accesses between iteration $i'$ and $i$. For each tuple $t \in \mathcal{T}$, we perform random access in order to retrieve the values of missing attributes (i.e., attributes of $t_\mathcal{Q}$ for which we do not know the values yet). Once the tuples in $\mathcal{T}$ are fully constructed, TA-SKY compares them against the tuples in the candidate skyline set. For each tuple $t \in \mathcal{T}$  three scenarios can arise:
\begin{enumerate}
    \itemsep0em
    \item $t$ dominates a tuple $t'$ in the tree (i.e., candidate skyline set), $t'$ is deleted from the tree.
    \item $t$ is dominated by a tuple $t'$ in the tree, it is discarded since it cannot be skyline.
    \item $t$ is not dominated by any tuple $t'$ in the tree, it is inserted in the tree.
\end{enumerate}

Once the candidate skyline set is updated with tuples in $\mathcal{T}$, we compare $t_{syn}$ with the tuples in the candidate skyline set. The algorithm stops when $t_{syn}$ is dominated by any tuple in the candidate skyline set.

We shall now explain TA-SKY for the subspace skyline query $\mathcal{Q}$ of Example~\ref{exmp:subspaceSkyline}. Sorted lists $\mathcal{L_Q}$ corresponding 
to query $\mathcal{Q}$ are shown in Figure~\ref{fig:sortedLists}. At iteration 1, TA-SKY retrieves the tuples $t_1, t_2$ and $ t_5$ by sequential access. For $t_1$ we know its value on attributes $A_2$ and $A_4$ whereas for $t_2$ and $t_5$ we know their value on $A_3$ and $A_1$ respectively. At this position we have $\mathcal{T} = \{ t_1, t_2, t_5 \}$ and $\tau_1 = \{1, 1, 1, 1\}$. Note that in addition to storing the tupleIDs that we have seen so far, we also keep track of the attribute values that are known from sequential access. After iteration 2,  $\mathcal{T} = \{ t_1, t_2, t_3, t_5, t_6\}$ and $\tau_2 = \{1, 1, 1, 1\}$. At iteration 3 we retrieve the values of $t_1, t_2, t_5$ and $ t_4$ on attributes $A_1, A_2, A_3,$ and $A_4$ respectively and update the corresponding entries $\mathcal{T}$.  Since $\tau_3 = \{0, 0, 1, 1\}$ is different from $\tau_2$, TA-SKY checks the stopping condition. First, we get the missing attribute values (attribute values which are not known from sequential access) of each tuple $t \in \mathcal{T}$. This is done performing random access on the appropriate sorted list in $\mathcal{L_Q}$. After all the tuples in $\mathcal{T}$ are fully constructed, we update the candidate skyline set using them. The final candidate skyline set is constructed after considering all the tuples in $\mathcal{T}$ is $\{t_1, t_5, t_6 \}$. Since the synthetic tuple $t_{syn} = \langle 0, 0, 1, 1 \rangle$ corresponds to $\tau_3$ is dominated by the candidate skyline set, we stop scanning the sorted lists and output the tuples in the candidate skyline set as the skyline answer set.

The number of tuples inserted into $\mathcal{T}$ (i.e., partially retrieved by sequential accesses) before the stopping condition is satisfied, impacts the performance of TA-SKY. This is because for each tuple $t \in \mathcal{T}$, we have to first perform random accesses in order to get the missing attribute values of $t$ and then compare $t$ with the tuples in the candidate skyline set in order to check if $t$ is skyline. Both the number of random accesses and number of dominance tests increase the execution time of TA-SKY. Hence, it is desirable to have a small number of entries in $\mathcal{T}$.  We noticed that the number of tuples inserted in $\mathcal{T}$ by TA-SKY depends on the organization of \textit{(tupleID, value)} pairs (i.e., ordering of pairs having same $value$) in sorted lists. Figure~\ref{fig:sortedListsOptimized} displays sorted lists $\mathcal{L'_Q}$ for the same relation in Example~\ref{exmp:subspaceSkyline} but with different organization. Both with $\mathcal{L_Q}$ and $\mathcal{L'_Q}$ TA-SKY stops at iteration 3. However, For $\mathcal{L_Q}$ after iteration 3, $\mathcal{T} = \{t_1, t_2, t_3, t_4, t_5, t_6\}$ and we need to make a total of 12 random accesses and 12 dominance tests\footnote{For each tuple $t \in \mathcal{T}$, we need to perform two dominance checks: i) if $t$ is dominating any tuple in the candidate skyline set and ii) if $t$ is dominated by tuples in the candidate skyline set.}. On the other hand, with $\mathcal{L'_Q}$, after iteration 3 we have $\mathcal{T} = \{t_1, t_2, t_5, t_6\}$, requiring only 4 random accesses and 8 dominance tests.

One possible approach to improve the performance of TA-SKY is to re-organize the sorted lists before running the algorithm for a given subspace skyline query. Specifically, $\forall t, t' \in D$ that $t[A_i] = t'[A_i]$, position $t$ before $t'$ in the sorted list $L_i$ $(1 \leq i \leq m')$ if $t$ has better value than $t'$ on the remaining attributes. However, re-arranging the sorted lists for each subspace skyline query will be costly. 

We now propose several optimization techniques that enable TA-SKY to compute skylines without considering all the entries in $\mathcal{T}$.

\vspace{1mm}
\noindent{\bf Selecting appropriate entries in $\mathcal{T}$:} Our goal is to only perform random access and dominance checks for tuples in $\mathcal{T}$ that are likely to be skyline for a given subspace skyline query. Consider a scenario where TA-SKY needs to check the stopping condition at iteration $k$, i.e, $\tau_k \neq \tau_{(k-1)}$. Let $\mathcal{Q'}$ be the set of attributes for which the value returned by sequential access at iteration $k$ is different from $(k-1)$-th iteration, $\mathcal{Q'} = \{A_i | A_i \in \mathcal{Q}, cv_{ki} < cv_{(k-1)i} \}$. In order for the tuple $t_{syn}$ to be dominated, there must exist a tuple $t' \in \mathcal{T}$ that has $t'[A_i] \geq t_{syn}[A_i]$, $\forall A_i \in \mathcal{Q}$ and $\exists A_i \in \mathcal{Q}$ $s.t.$ $t'[A_i] > t_{syn}[A_i]$. Note that each tuple $t \in \mathcal{T}$ has $t[A_i] = t_{syn}[A_i], \forall A_i \in \mathcal{Q} \setminus \mathcal{Q'}$. This is because for all $A_i \in \mathcal{Q} \setminus \mathcal{Q'}$ sorted access returns same value on both $(k-1)$-th and $k$-th iteration (i.e., $cv_{(k-1)i} = cv_{ki}$). Hence, the only way a tuple $t' \in \mathcal{T}$ can dominate $t_{syn}$ is to have a larger value on any of the attributes in $\mathcal{Q'}$. Therefore, we only need to consider a subset of tuples $\mathcal{T'} = \{ t | t \in \mathcal{T}, \exists A_i \in \mathcal{Q} \setminus \mathcal{Q} \text{ s.t. } t[A_i] = cv_{(k-1)i} \}$. Note that it is still possible that $\exists t, t' \in \mathcal{T'} \text{ s.t. } t \succ_{\mathcal{Q}} t'$. Thus, we need to only consider the tuples that are skylines among $\mathcal{T'}$ and the candidate skyline set. To summarize, before checking the stopping condition at iteration $k$, we have to perform the following  operations: (i) Select a subset of tuples $\mathcal{T'}$ from $\mathcal{T}$ that are likely to dominate $t_{syn}$, (ii) For each tuple $t \in \mathcal{T}$ get the missing attribute values of $t$ performing random access on appropriate sorted lists, (iii) Update the candidate skyline set using the skylines in $\mathcal{T'}$, and (iv) Check if $t_{syn}$ is dominated by the updated candidate skyline set.

Note that in addition to reducing the number of random access and dominance test, the above optimization technique makes the TA-SKY algorithm {\em progressive}, i.e, tuples that are inserted into the candidate skyline set will always be skyline in the query space $\mathcal{Q}$.
This characteristic of TA-SKY makes it suitable for real-world web applications where instead of waiting for all the results to be returned users want a subset of the results very quickly.


\vspace{1mm}
\noindent{\bf Utilizing the ST algorithms:} We can utilize the ST algorithms for discovering the skyline tuples from $\mathcal{T'}$. This way we can take advantages of the optimization approaches proposed in \S\ref{sec:3}. For example, we can call ST-S algorithm with parameter: tree $T$ (stores all the tuples discovered so far) and tuple list $\mathcal{T'}$. The output skyline tuples in  $\mathcal{T'}$ that are not dominated by $T$. Moreover, after sorting the tuples in ST-S, if we identify that $score(t_i) = score(t_{i-1})$ $(2 \leq i \leq |\mathcal{T'}|)$ and $t_{i-1}$ is dominated, we can safely mark $t_i$ as dominated. This is because $score(t_i) = score(t_{i-1})$ implies that both $t_i$ and $t_{i-1}$ have same attribute value assignment. When the number of attributes in a subspace skyline query is small, this approach allows us to skip a large number of dominance tests.

The pseudocode of TA-SKY, after applying the optimizations above, is presented in Algorithm~\ref{alg:taSky}.

\begin{algorithm}[!htb]
\caption{{\bf TA-SKY}}
\begin{algorithmic}[1]
\label{alg:taSky}
\STATE {\bf Input:} Query $\mathcal{Q}$, Sorted lists $\mathcal{L_Q}$; \\ {\bf Output:} $\mathcal{S}_\mathcal{Q}$.
\STATE $T = $ {\it New} Tree(); $\mathcal{T} = \emptyset$
\STATE {\bf repeat}
    \STATE \hindent $\tau = \emptyset$
    \STATE \hindent {\bf for} each sorted list $L_i \in \mathcal{L_Q}$
        \STATE \hindent[2] $A_i$ = Attribute corresponds to $L_i$
        \STATE \hindent[2] $(tupleID, value) = SortedAccess(L)$
        \STATE \hindent[2] $\mathcal{T}[tupleID][A_i] = value$
        \STATE \hindent[2] $\tau[A_i] = value$
    \STATE \hindent {\bf if} $\tau$ remains unchanged from prev. iteration:
        \STATE \hindent[2] {\bf continue;}
    \STATE \hindent $\mathcal{Q'} = \{A_i | A_i \in \mathcal{Q}, \tau[A_i] \text{ changed from prev.iteration}\}$
    \STATE \hindent $\mathcal{T'} = \{t | t \in \mathcal{T}, \exists A_i \in \mathcal{Q'}, \mathcal{T}[t][A_i] \text{ is set} \}$
    \STATE \hindent Delete entries from $\mathcal{T}$ that are inserted in $\mathcal{T'}$
    \STATE \hindent {\bf for} each $t \in \mathcal{T'}$
        \STATE \hindent[2] {\bf for} each attribute $A_i \in Q \setminus Q'$
            \STATE \hindent[3] {\bf if} $t[A_i]$ is missing:
                \STATE \hindent[4] $t[A_i]= RandomAccess(L, A_i)$
        
        \STATE \hindent[2] Update $score$ of $t$
    \STATE \hindent ST-S($\mathcal{T}$, $\mathcal{Q}$, $T$)
    \STATE \hindent $t_{syn}=$ Synthetic tuple with values of $\tau$
\STATE {\bf until} IS-DOMINATED($t_{syn}$, $T.root$, 1, $score(t_{syn})$)
\end{algorithmic}
\end{algorithm}

\subsubsection{Performance Analysis}\label{sec:TASKY-performance}

\vspace{1mm}
\noindent{\bf Worst Case Analysis:} In the worst case, TA-SKY will exhaust all the $m^\prime$ sorted lists. Hence, will perform $O(m^\prime n)$ sorted and $O(m^\prime n)$ random accesses. After all the tuples are fully constructed, for each tuple $t$, we need to check whether any other tuple in $T$ dominates $t$. The cost of each dominance check operation is $O(m^\prime n)$. Hence, cost of $n$ dominance checks is $O(m^\prime n^2)$. Therefore, the worst case time complexity of TA-SKY is $O(m'n^2)$

\vspace{1mm}
\noindent{\bf Expected Cost Analysis:}

\begin{lemma}\label{lemma:expectedDiscovery}
Considering $p_i$ as the probability that a tuple has value 1 on the binary attribute $A_i$, the expected number of tuples discovered by TA-SKY after $i$ iterations is:
\begin{align}\label{eq:expectedDiscovery}
n P_{seen}(t,i)
\end{align}
where $P_{seen}(t,i)$ is computed using Equation~\ref{eq:seen}.
\begin{align}\label{eq:seen}
\nonumber
P_{seen}(t,i) = 1 - \prod_{j=1}^{m^\prime} \Big( & (1 - p_j) \big(\sum_{k=0}^{i-1}P_{L_j}(k)\frac{n-i}{n-k} + \sum_{k=i}^{n}P_{L_j} \big) \\
                                                  & + p_j\sum_{k=i+1}^n P_{L_j}(k) 
\end{align}
\end{lemma}
Refer to Appendix~\ref{sec:appendixProof} for the proof.

\begin{theorem}\label{thm:expectedCostTA-SKY}
Given a subspace skyline query $\mathcal{Q}$, the expected number of sorted accesses performed by TA-SKY on an $n$ tuple boolean relation with probability of having value $1$ on attribute $A_j$ being $p_j$ is,
\begin{align}
m^\prime \sum_{i=1}^n i\times P_{stop}(i)
\end{align}
where $P_{stop}(i)$ is computed using Equations~\ref{eq:stopi-1},~\ref{eq:stopi-2}, and~\ref{eq:stopi-3}.
\begin{align}\label{eq:stopi-1}
P_{stop}(i) &= \sum_{k=1}^m P_0(i, k) \times {m^\prime \choose k} \times (1 - (1 - P_{stop}(t, \mathcal{Q}_k))^{i^\prime})
\end{align}
\begin{align}\label{eq:stopi-2}
P_0(i, k) = {m^\prime \choose k} \prod_{A_j \in \mathcal{Q}_k} (1 - p_j)^{n-i} \prod_{A_j \in \mathcal{Q} \setminus \mathcal{Q}_k} \big(1 - (1 - p_j)^{n-i}\big)
\end{align}
\begin{align}\label{eq:stopi-3}
P_{stop}(t, \mathcal{Q}_k) &= \underset{\forall A_j \in \mathcal{Q}\backslash \mathcal{Q}_k} {\Pi} p_j (1-\underset{\forall A_j \in \mathcal{Q}_k} {\Pi} (1 - p_j))
\end{align}
\end{theorem}
The proof is available in Appendix~\ref{sec:appendixProof}

\section{Related Work}\label{sec:relWork}

%
%
%
%

In the database context, the skyline operator was first introduced in \cite{borzsony2001skyline}. Since then much work aims to improve the performance of skyline computation in different scenarios. In this paper, we consider skyline algorithms designed for centralized database systems. 

To the best of our knowledge, LS~\cite{morse2007efficient} and Hexagon~\cite{preisinger2007hexagon} are the only two algorithms designed to compute skylines over categorical attributes. Both algorithms operate by first creating the complete lattice of possible attribute-value combinations. Using the lattice structure, non-skyline tuples are then discarded. Even though LS and Hexagon can discover the skylines in linear time, the requirement to construct the entire lattice for each skyline is strict and not scalable. The size of the lattice is exponential in the number of attributes in a skyline query. Moreover, in order to discover the skylines, the algorithms have to scan the entire dataset twice, which is not ideal for online applications.

Most of the existing work on skyline computation concerns relations with {\em numeric attributes}. Broadly speaking, skyline algorithms for numerical attributes can be categorized as follows. {\em Sorting-based Algorithms} utilize sorting to improve the performance of skyline computation aiming to discard nonskyline objects using a small number of dominance checks~\cite{chomicki2005skyline}~\cite{godfrey2005maximal}. For any subspace skyline query, such approaches will require sorting the dataset. SaLSa~\cite{bartolini2008efficient} is the best in this category and we demonstrated how our adaptation on categorical domains, namely ST-S outperforms SaLSa. 

{\em Partition-based Algorithms} recursively partition the dataset into a set of disjoint regions, compute local skylines for each region and merge the results 
\cite{borzsony2001skyline}~\cite{zhang2009scalable}. Among these, BSkyTree~\cite{lee2014scalable} has been shown to be the best performer. We
demonstrated that our adaptation of this algorithm, namely ST-P, for categorical domains outperforms the vanilla BSkyTree when applied to our application scenario. Other partitioning algorithms, such as  NN~\cite{kossmann2002shooting}, BBS~\cite{papadias2003optimal} 
and ZSearch~\cite{lee2007approaching} utilize indexing structures such as R-tree, ZB-tree for efficient region level dominance tests. However,
adaptations of such algorithms in the subspace skyline problem would incur exponential space overhead which is not in line with the scope
of our work (at most linear to the number of attributes overhead). 

A body of work is also devoted to {\em Subspace Skyline Algorithms}~\cite{yuan2005efficient, pei2005catching} which utilize pre-computation to
compute skylines for each subspace skyline query. These algorithms impose exponential space overhead, however. Further improvements
to reduce the overhead ~\cite{tao2006subsky}~\cite{xia2006refreshing}~\cite{xia2012online}~\cite{maabout2016skycube} are highly data
dependent and offer no guarantees for their storage requirements.

\begin{figure*}[!ht]
  \begin{minipage}[t]{0.23\linewidth}
    \centering
    \includegraphics[scale=0.46]{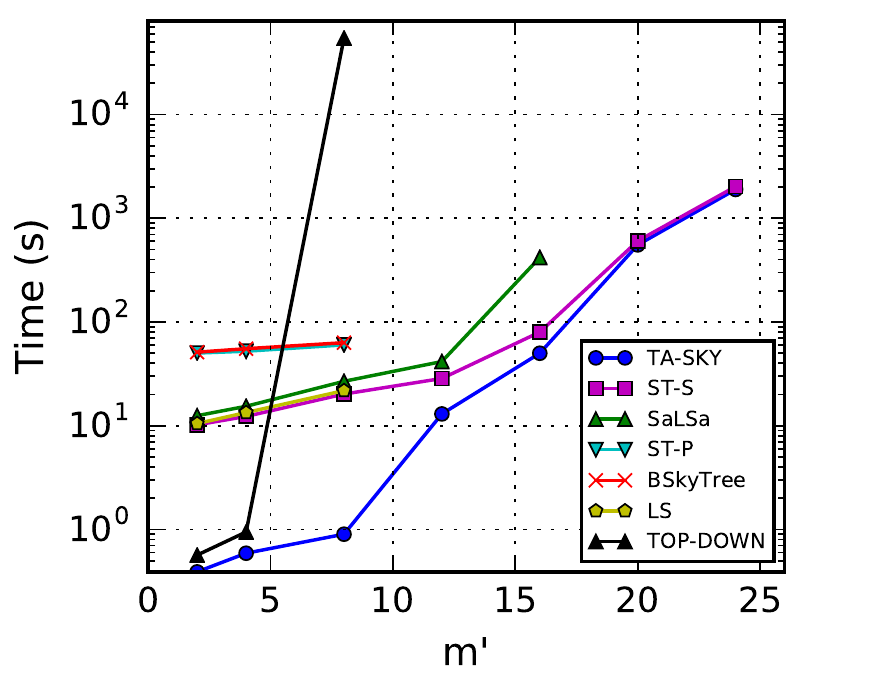}
    \vspace{-6mm}
    \caption{Varying query size}
    \label{fig:algorithms}
  \end{minipage}
  \hspace{1mm}
  \begin{minipage}[t]{0.23\linewidth}
    \centering
    \includegraphics[scale=0.46]{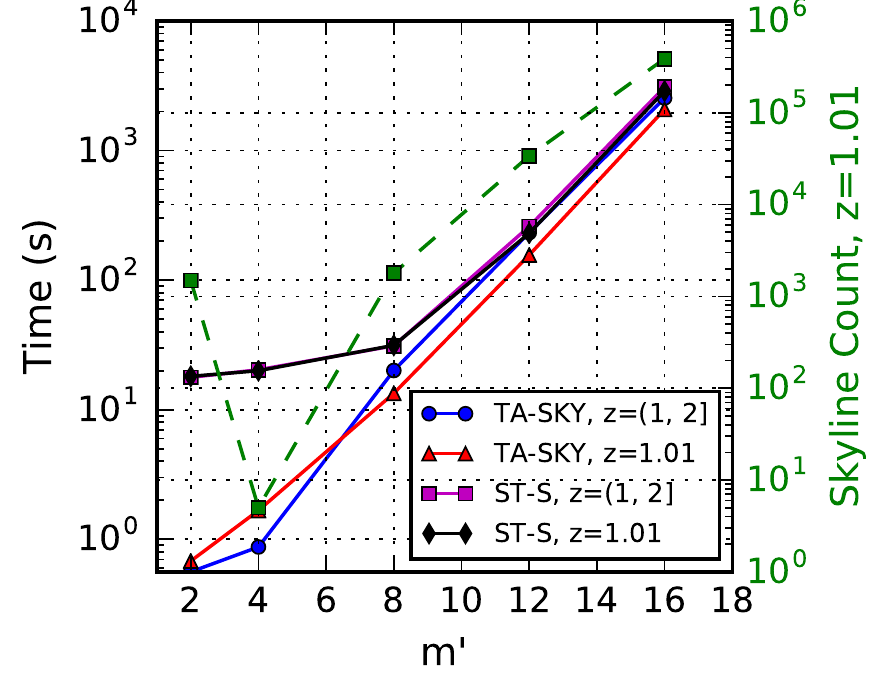}
    \vspace{-6mm}
    \caption{Varying query size}
    \label{fig:syn_TimeVsDimension_Dist}
  \end{minipage}
  \hspace{3mm}
  \begin{minipage}[t]{0.23\linewidth}
    \centering
    \includegraphics[scale=0.46]{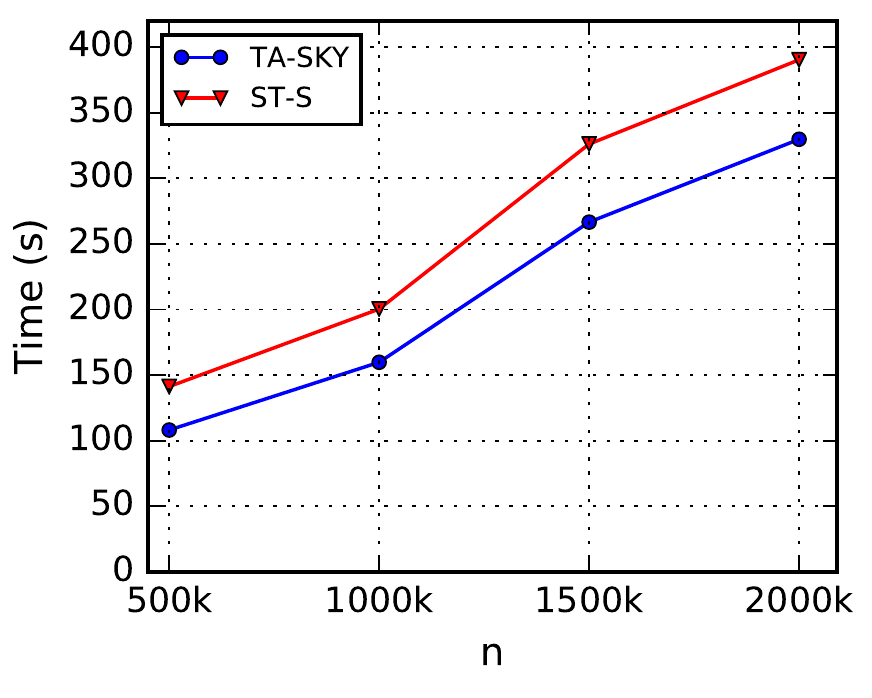}
    \vspace{-6mm}
    \caption{Varying number of tuples} 
    \label{fig:syn_TimeVsN}
  \end{minipage}
  \begin{minipage}[t]{0.25\linewidth}
    \centering
    \includegraphics[scale=0.46]{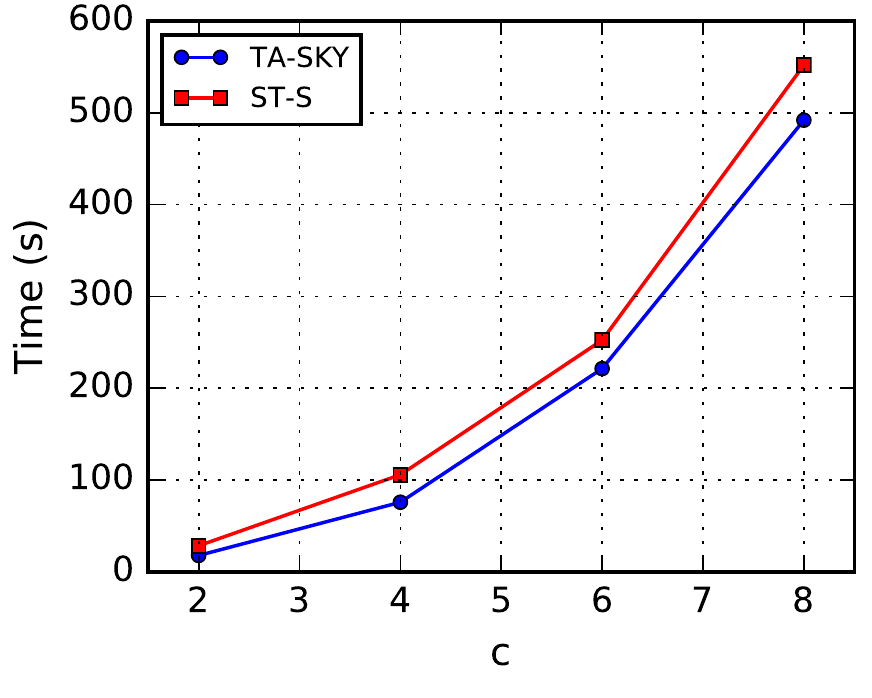}
    \vspace{-2mm}
    \caption{Varying cardinality}
    \label{fig:syn_TimeVsC}
  \end{minipage}
\end{figure*}

\begin{figure*}[!ht]
  \begin{minipage}[t]{0.25\linewidth}
    \centering
    \includegraphics[scale=0.46]{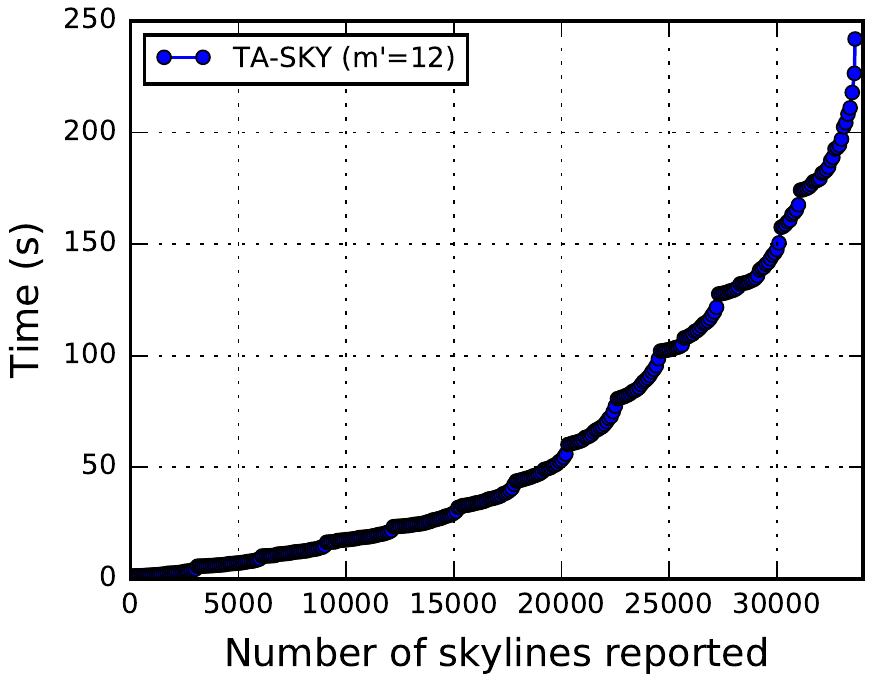}
    \caption{Time vs number of skylines returned}
    \label{fig:syn_TimeVsNumberOfSkylines}
  \end{minipage}
  \hspace{2mm}
  \begin{minipage}[t]{0.25\linewidth}
    \centering
    \includegraphics[scale=0.46]{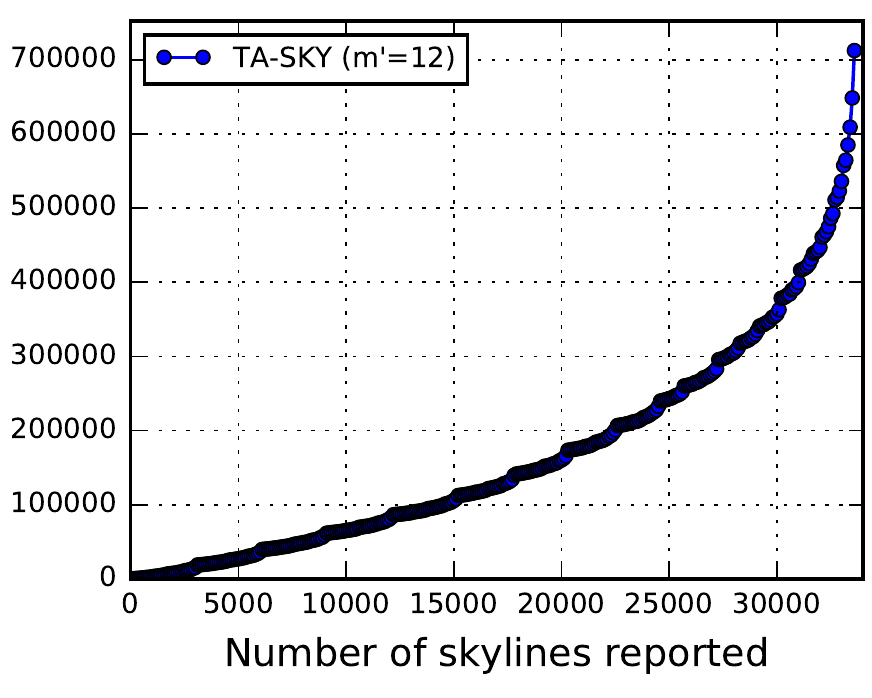}
    \caption{Tuples accessed vs number of skylines returned}
    \label{fig:syn_NumberOfTuplesVsSkylines}
  \end{minipage}
  \begin{minipage}[t]{0.21\linewidth}
    \centering
    \includegraphics[scale=0.5]{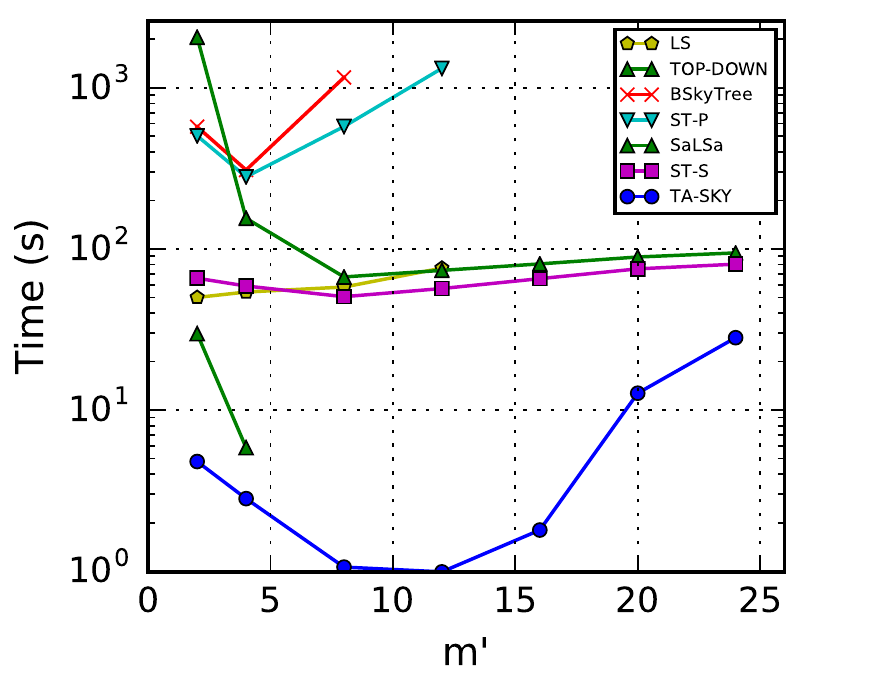}
    \caption{AirBnB: Varying query size}
    \label{fig:algorithmsAirbnb}
  \end{minipage}
  \hspace{1mm}
  \begin{minipage}[t]{0.23\linewidth}
    \centering
    \includegraphics[scale=0.46]{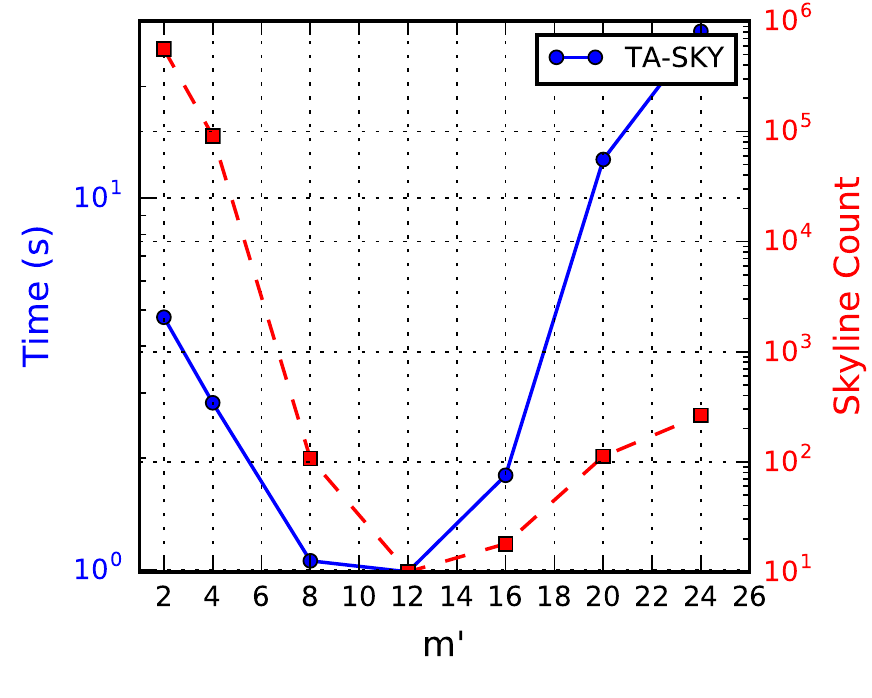}
    \caption{AirBnB: TA-SKY performance v.s. Skyline size}
    \label{fig:Airbnbm}
  \end{minipage}
\end{figure*}

\section{Experimental Evaluation}\label{sec:experiments}
\subsection{Experimental Setup}

In this section, we describe our experimental results. In addition to the theoretical analysis presented in \S\ref{sec:3} and \S\ref{sec:subsky}, we compared our algorithms experimentally against existing state-of-the-art algorithms.  Our experiments were run over synthetic data, as well as two real-world datasets collected from \emph{AirBnB}\footnote{\small{http://www.airbnb.com/}} and \emph{Zillow}\footnote{\small{http://zillow.com/}}.

\vspace{1mm}
\noindent {\bf Synthetic Datasets:} 
In order to study the performance of the proposed algorithms in different scenarios, we generated a number of {\bf Zipfian datasets}, each containing 2M tuples and 30 attributes. Specifically, we created datasets with attribute cardinality ranging from  $2 - 8$. In this environment, the frequency of an attribute value is inversely proportional to its rank.  Therefore, the number of tuples having a higher (i.e., better) attribute value is less than then number of tuples with a comparatively lower attribute value. We used a Python package for generating these datasets. For each attribute, we specify its distribution over the corresponding domain by controlling the $z$ value. Two attributes having the same cardinality but different $z$ values will have different distributions. Specifically, the attribute with lower $z$ value will have a higher number of tuples having higher attribute value. Unless otherwise specified, we set the $z$ values of the attributes evenly distributed in the range $(1, 2]$ for generating synthetic datasets.

\vspace{1mm}
{\bf Choice of dataset:} we used Zipfian datasets as they reflect more precisely situation with real categorical datasets. Specifically, in real-world applications, for a specific attribute, the number of objects having higher attribute values (i.e., better) is likely to be less than the number of objects with lower attribute values. For example, in AirBnB, \emph{3 bed room} hosts are less frequent than hosts having a \emph{single bed room}. Similarly, in Craigslist, \emph{sedans} are more prevalent than \emph{sports cars}. Moreover, in real-world applications, the distributions of attributes are different from one another. For example, in our AirBnB dataset, approximately 600k out of the 2M hosts have amenity \emph{Cable TV}. Whereas, the approximate number of hosts with amenity \emph{Hot Tub} is only 200k.

\vspace{1mm}
\noindent {\bf AirBnB Dataset:} Probably one of the best fits for the application of this paper is AirBnB. It is a peer-to-peer location-based marketplace in which people can rent their properties or look for an abode for a temporary stay. We collected the information of approximately 2 million \emph{real} properties around the globe, shared on this website. AirBnB has a total number of 41 attributes for each host that captures the features and amenities provided by the hosts. Among all the attributes, 36 of them are boolean (categorical with domain size 2) attributes, such as \emph{Breakfast}, \emph{Cable TV}, \emph{Gym}, and \emph{Internet}, while 5 are categorical attributes, such as \emph{Number of Bedrooms}, and \emph{Number of Beds} etc. We tested our proposed algorithms against this dataset to see their performance on real-world applications.

\vspace{1mm}
\noindent {\bf Zillow Dataset:}
Zillow is a popular online real estate website that helps users to find houses and apartments for sale/rent. We crawled approximately 240k houses listed for sale in Texas and Florida state. For each listing, we collected 9 attributes that are present in all the houses. Out of 9 attributes, 7 of them are categorical, such as \emph{House Type}, \emph{Number of Beds}, \emph{Number of Baths}, \emph{Parking Space} etc., and two are numeric - \emph{House size} (in sqft), and \emph{Price}. The domain cardinalities of the categorical attributes varies from 3 to 30. Using discretization we mapped the numeric attributes into the categorical domain, each of cardinality 20.

\vspace{1mm}
\noindent {\bf Algorithms Evaluated:}
We tested the proposed algorithms, namely ST-S, ST-P, TOP-DOWN, and TA-SKY as well as the state-of-art algorithms LS~\cite{morse2007efficient}, SaLSa~\cite{bartolini2008efficient} and BSkyTree \cite{lee2014scalable} that are applicable to our problem settings.

\vspace{1mm}
\noindent {\bf Performance Measures:} We consider running time as the main performance measure of the algorithms proposed in this paper. In addition, we also investigate the key features of ST-S, ST-P and TA-SKY algorithm and demonstrate how they behave under a variety of settings. Each data point is obtained as the average of 25 runs.

\vspace{1mm}
\noindent{\bf Hardware and Platform:} All our experiments were performed on a quad-core 3.5 GHz Intel i7 machine running Ubuntu 14.04 with 16 GB of RAM. The algorithms were implemented in Python.

\subsection{Experiments over Synthetic Datasets}

\vspace{1mm}
\noindent{\bf Effect of Query Size $m^\prime$} \textbf{:} We start by comparing the performance of our algorithms with existing state-of-art algorithms that exhibit the best performance in their respective domain. Note that, unlike TA-SKY, the rest of the algorithms do not leverage any indexing structure. The goal of this experiment is to demonstrate how utilizing a small amount of precomputation (compared to the inordinate amount of space required by Skycube algorithms) can improve the performance of subspace skyline computation. Moreover, the precomputation cost is independent of the skyline query. This is because we only need to build the sorted lists once at the beginning. For this experiment, we set $n=500$k and vary $m'$ between $6-24$. In order to match real-world scenarios, we selected attributes with cardinality $c$ ranging between $2-6$. Specifically, $50\%$ of the selected attributes have cardinality $2$, $30\%$ have cardinality $4$, and $20\%$ have cardinality $6$. Figure~\ref{fig:algorithms} shows the experiment result. We can see that when $m^\prime$ is small, TA-SKY outperforms other algorithms. This is because, with small query size, TA-SKY can discover all the skylines by accessing only a small portion of the tuples in the dataset. However as $m'$ increases, the likelihood of a tuple dominating another tuple decreases. Hence, the total number of tuples accessed by TA-SKY before the stopping condition is satisfied also increases. Hence, the performance gap between TA-SKY and ST-S starts to decrease. Both ST-S and ST-P exhibits better performance compared to their baseline algorithms (SaLSa and BSkyTree). Algorithms such as ST-P, BSkyTree, and LS do not scale for larger values of $m'$. This is because all these algorithms operate by constructing a lattice over the query space which grows exponentially. Moreover, even though TOP-DOWN initially performed well, it did not not complete successfully for $m^\prime > 4$.


Figure~\ref{fig:syn_TimeVsDimension_Dist} demonstrates the effect $m^\prime$ and $z$ on the performance of TA-SKY and ST-S. For this experiment, we created two datasets with cardinality $c = 6$ and different $z$ values. In the first dataset, all the attributes have same $z$ value (i.e., $z=1.01$), whereas, for the second dataset, $z$ values of the attributes are evenly distributed within the range $(1, 2]$. By setting $z = 1.01$ for all attributes, we increase the frequency of tuples having preferable (i.e., higher) attribute values. Hence, the skyline size of the first dataset is less than the skyline size of the second dataset. This is because tuples with preferable attribute values are likely to dominate more non-skyline tuples, resulting in a small skyline size. Moreover, this also increases the likelihood of the stopping condition being satisfied at an early stage of the iteration. Hence, TA-SKY needs less time for the dataset with $z = 1.01$. In summary, TA-SKY performs better on datasets where more tuples have preferable attribute values.
The right-y-axis of Figure~\ref{fig:syn_TimeVsDimension_Dist} shows the skyline size for each query length. One can see that as the query size increased, the chance of tuples dominating each other decreased, which resulted in a significant increase in the skyline size. Please note that the increases in the execution time of TA-SKY are due to the increase in the skyline size which is bounded by $n$. Moreover, as $m^\prime$ increases, there is an initial decrease in skyline size. This is because when $m^\prime$ is small (i.e., 2), the likelihood of a tuple having highest value (i.e., preferable) on all attribute is large.

\vspace{1mm}
\noindent{\bf Effect of Dataset Size ($n$):} Figure~\ref{fig:syn_TimeVsN} shows the impact of $n$ on the performance of TA-SKY and ST-S. For this experiment, we used dataset with cardinality $c=6$, $m^\prime=12$ and varied $n$ from $500$K to $2$M. As we increase the value of $n$, the number of skyline tuples increases. With the increase of skyline size, both TA-SKY and ST-S needs to process more tuple before satisfying the stop condition. Therefore, total execution time increases with the increase of $n$.

\vspace{1mm}
\noindent{\bf Effect of Attribute Cardinality ($c$):} In our next experiment, we investigate how changing attribute cardinality affects the execution time of TA-SKY and ST-S. We set the dataset size to $n=1$M while setting the query size to $m^\prime = 12$, and vary the attribute cardinality $c$ from $4$ to $8$. Figure~\ref{fig:syn_TimeVsC} shows the experiment result. Increasing the cardinality of the attributes increases the total number of skyline tuples. Therefore, effects the total execution time of TA-SKY and ST-S.

\begin{figure*}[!ht]
  \vspace{-6mm}
  \begin{minipage}[t]{0.23\linewidth}
    \centering
    \includegraphics[scale=0.46]{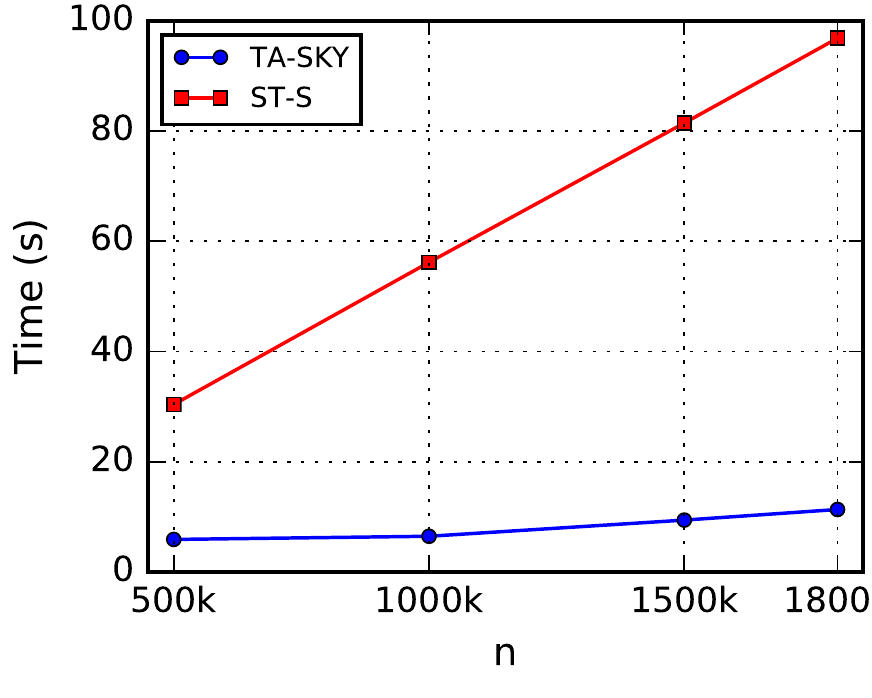}
    \caption{AirBnB: Varying the number of tuples}
    \label{fig:Airbnbn}    
  \end{minipage}
  \begin{minipage}[t]{0.23\linewidth}
    \centering
    \includegraphics[scale=0.45]{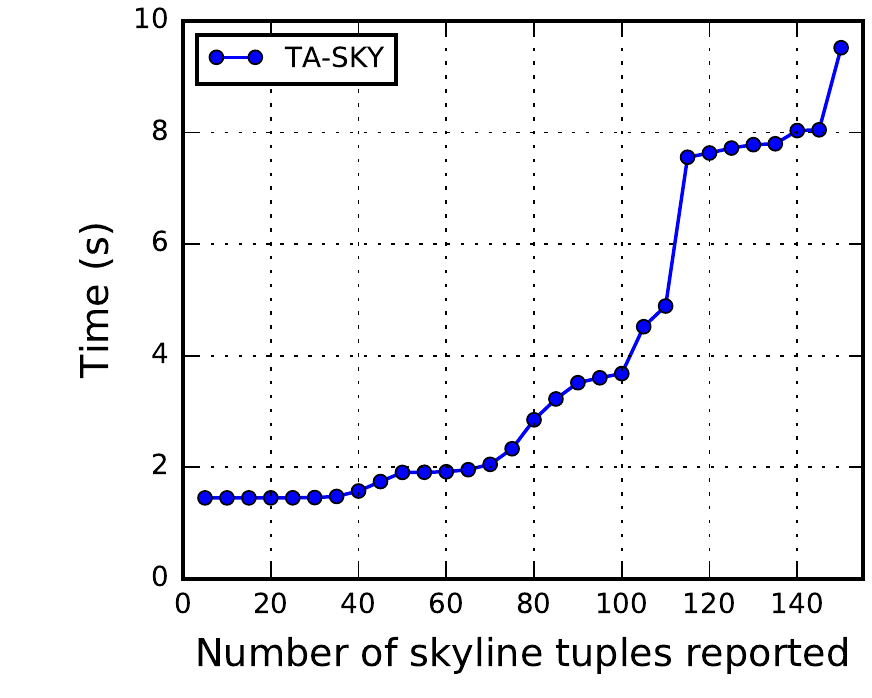}
    \caption{AirBnB: Time vs the number of skylines returned} 
    \label{fig:Airbnbp1}
  \end{minipage}
  \hspace{1mm}
  \begin{minipage}[t]{0.25\linewidth}
    \centering
    \includegraphics[scale=0.45]{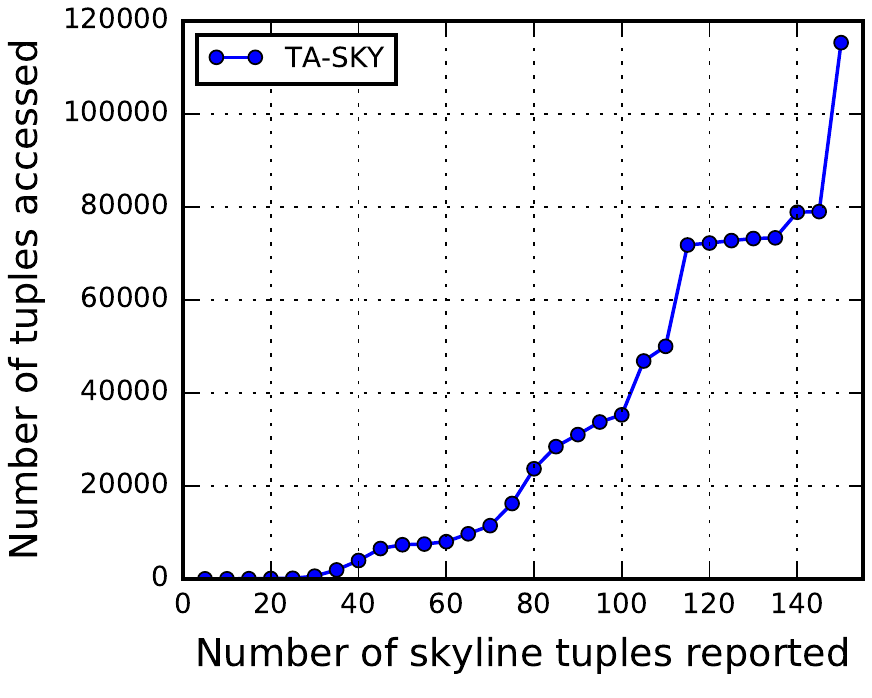}
    \caption{AirBnB: Number of accessed tuples vs the number of skylines}
    \label{fig:Airbnbp2}
  \end{minipage}
  \begin{minipage}[t]{0.25\linewidth}
    \centering
    \includegraphics[scale=0.45]{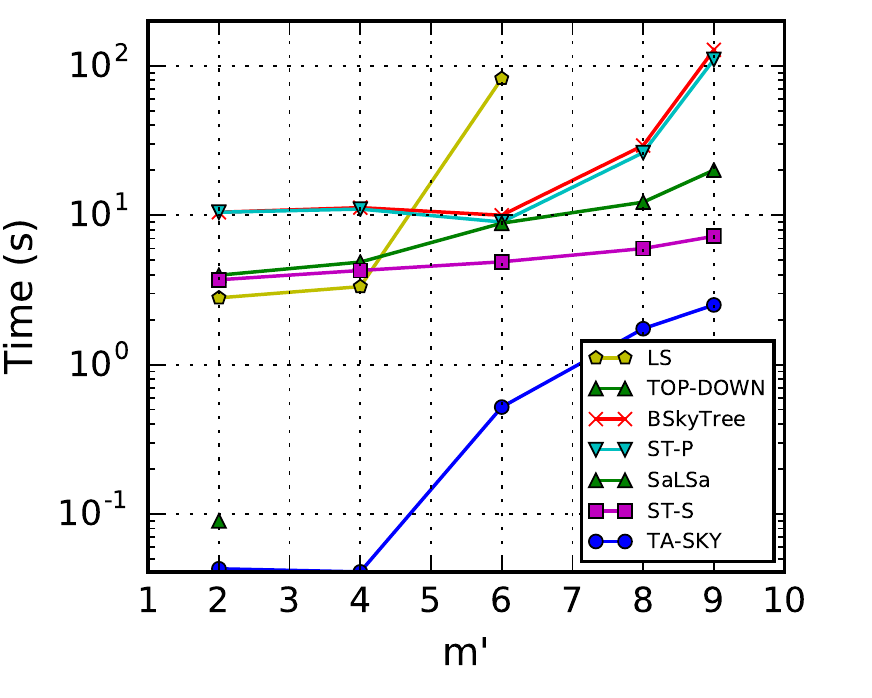}
    \caption{Zillow: Varying query size}
    \label{fig:Zillowm}
  \end{minipage}
\end{figure*}

\begin{figure*}[!ht]
  \begin{minipage}[t]{0.3\linewidth}
    \centering
    \includegraphics[scale=0.46]{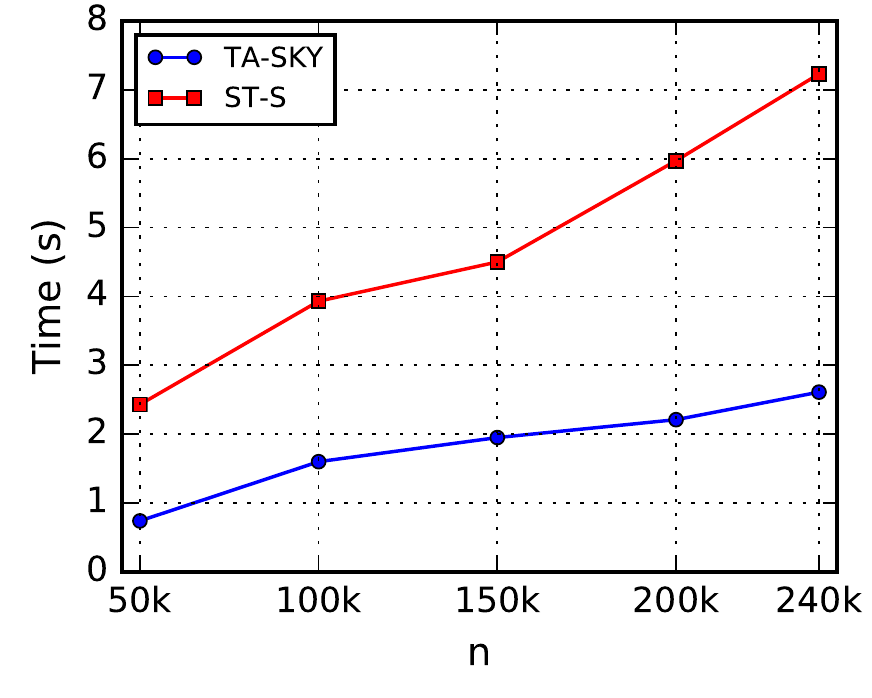}
    \caption{Zillow: Varying the number of tuples}
    \label{fig:TimeVsNZillow}
  \end{minipage}
  \hspace{2mm}
  \begin{minipage}[t]{0.3\linewidth}
    \centering
    \includegraphics[scale=0.46]{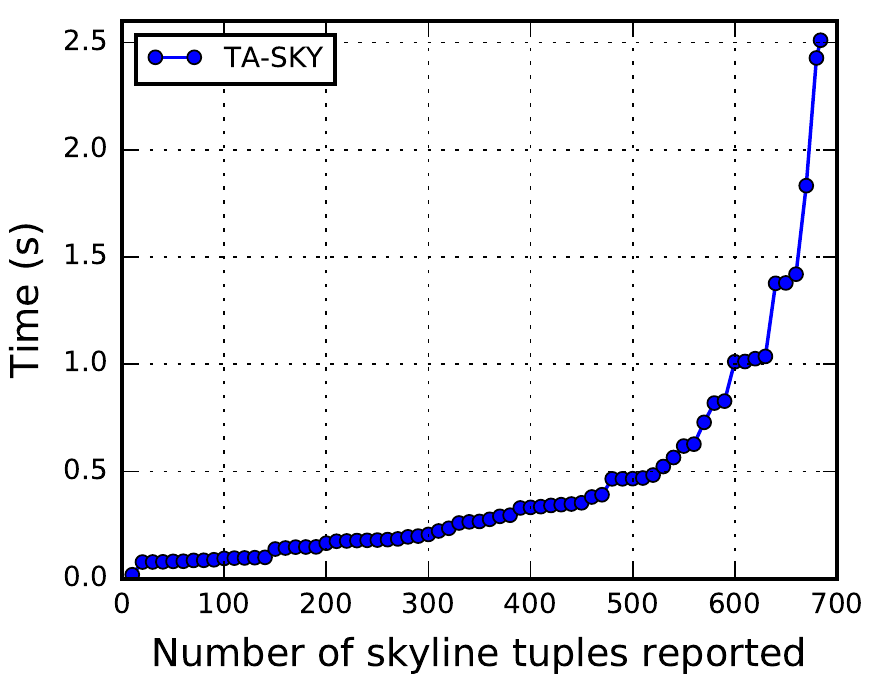}
    \caption{Zillow: Time vs the number of skylines returned}
    \label{fig:NumberOfTuplesVsSkylines}
  \end{minipage}
  \begin{minipage}[t]{0.3\linewidth}
    \centering
    \includegraphics[scale=0.5]{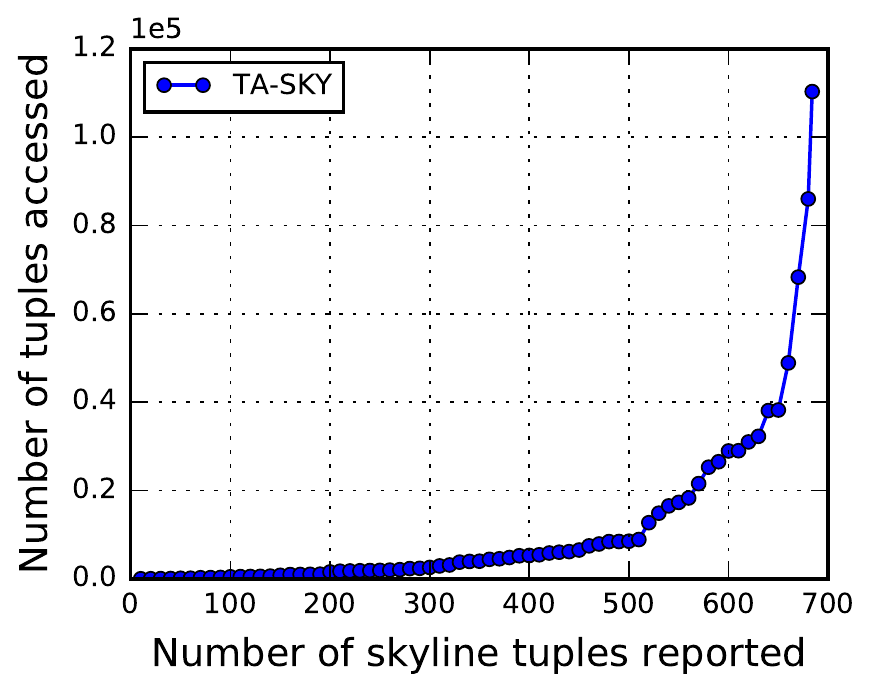}
    \caption{Zillow: Number of accessed tuples vs the number of skylines}
    \label{fig:NumberOfTuplesVsOutputTuplesZillow}
  \end{minipage}
\end{figure*}

\vspace{1mm}
\noindent{\bf Progressive Behavior of TA-SKY:} Figure~\ref{fig:syn_TimeVsNumberOfSkylines} and~\ref{fig:syn_NumberOfTuplesVsSkylines} demonstrates the incremental performance of TA-SKY for discovering the new skylines for a specific query of size $m^\prime=12$, while $n=1$M and all the attributes having cardinality $c=12$. Figure~\ref{fig:syn_TimeVsNumberOfSkylines} shows the CPU time as a function of the skyline size returned. We can see that even though the full skyline discovery takes $250$ seconds, within the first $50$ seconds TA-SKY outputs more than $50\%$  of the skyline tuples. Figure~\ref{fig:syn_NumberOfTuplesVsSkylines} presents the number of tuples TA-SKY accessed as a function of skyline tuples discovered so far. The skyline contains more than $33$k tuples. In order to discover all the skylines, TA-SKY needs to access almost $700$K (70\%) tuples. However, we can see that more than $80$\% of the skyline tuples can be discovered by accessing less that $30$\% tuples.

\subsection{Experiments over AirBnB Dataset}\label{subsec:expAirbnb}
In this experiment, we test the performance of our final algorithm, TA-SKY, against the real Airbnb dataset. We especially study (i) the effects of varying $m'$ and $n$ on the performance of the algorithm and (ii) the progressive behavior of it.

\vspace{1mm}
\noindent{\bf Effect of Varying Query Size ($m'$):} In our first experiment on AirBnB dataset, we compared the performance of different algorithms proposed in the paper with existing works.
We varied the number of attributes in the query (i.e., $m'$) from $2$ to $24$ while setting the number of tuples to 1,800,000.
Figure~\ref{fig:algorithmsAirbnb} shows the experiment result.
Similar to our experiment on the synthetic dataset (Figure~\ref{fig:algorithms}), TA-SKY and ST-S perform better than the remaining algorithms.
Even though initially performing well, TOP-DOWN did not scale after query length $4$. This is because, with $m^\prime > 4$, the skyline hosts shift to the middle of the corresponding query lattice, requiring TOP-DOWN to query many lattice nodes.
Figure~\ref{fig:Airbnbm} shows the relation between the performance of TA-SKY and the skyline size.
Unlike the generally accepted rule of thumb that the skyline size grows exponentially as the number of attributes increases, in this experiment, we see that the skyline size originally started to decrease as the query size increased and then started to increase again after query size $12$. The reason for that is because when the query size is small and $n$ is relatively large, the chance of having many tuples with (almost) all attributes in $\mathcal{Q}$ being 1 (for Boolean attributes) is high. None of these tuples are dominated and form the skyline. However, as the query size increases, the likelihood of having a tuple in the dataset that corresponds to the top node of the lattice decreases. Hence, if the query size gets sufficiently large, we will not see any tuple corresponding to the top node. From then the skyline size will increase with the increase of query size.

\vspace{1mm}
\noindent{\bf Effect of Varying Dataset Size ($n$):}
In this experiment, we varied the dataset size from 500,000 to 1,800,000 tuples, while setting $m'$ to $20$.
Figure~\ref{fig:Airbnbn} shows the performance of TA-SKY and ST-S in this case. Once can see that between these two algorithms, the cost of ST-S grows faster. Moreover, even though in the worst case TA-SKY is quadratically dependent on $n$, it performs significantly better in practice. Especially in this experiment, a factor of 4 increase in the dataset size only increased the execution time by less than a factor of 3.

\vspace{1mm}
\noindent{\bf Progressive Behavior of TA-SKY:}
As explained in \S\ref{sec:TASky}, TA-SKY is a progressive algorithm, i.e., tuples that are inserted into the candidate skyline set are guaranteed to be in $\mathcal{S}_\mathcal{Q}$.
This characteristic of TA-SKY makes it suitable for real world (especially web) applications, where, rather than delaying the result until the algorithm ends, partial results can gradually be returned to the user. Moreover, we can see that TA-SKY tends to discover a large portion of the skyline quickly within a short execution time with a few number of tuple accesses (as a measure of cost in the web applications).
To study this property of the algorithm, in this experiment, we set $n=1,800,000$ and $m'=20$ and monitored the execution time, as well as the number of tuple accesses, as the new skyline tuples are discovered. 
Figures~\ref{fig:Airbnbp1} and~\ref{fig:Airbnbp2} show the experiment results for the execution time and the number of accessed tuples, respectively.
One can see in the figure that TA-SKY performed well in discovering a large number of tuples quickly. For example, (i) as shown in Figure~\ref{fig:Airbnbp1}, it discovered more than $\frac{2}{3}$ of the skylines in less that $3$ seconds, and (ii) as shown in Figure~\ref{fig:Airbnbp2}, more than half of the skylines were discovered by only accessing less than $2\%$ of the tuples ($20,000$ tuples).

\subsection{Experiments over Zillow Dataset}\label{subsec:expZillow}

We performed the similar set of experiments on Zillow dataset.

\vspace{1mm}
In our first experiment, we varied the number of attributes from 2 to 9 while the $n$ is set to 236,194. The experiment result is presented in Figure \ref{fig:Zillowm}. Similar to our previous experiments, ST-S and TA-SKY outperforms the remaining algorithms. This result also shows the effectiveness of ST-S and TA-SKY on categorical attributes with large domain size. For the next experiment, we varied the dataset size ($n$) from 50,000 to 240,000 tuples, while setting $m^\prime$ to 9. Figure \ref{fig:TimeVsNZillow} shows the performance of ST and TA-SKY for this experiment. Figure \ref{fig:NumberOfTuplesVsSkylines} and \ref{fig:NumberOfTuplesVsOutputTuplesZillow} demonstrate the progressive behavior of TA-SKY for $m^\prime = 9$ and $n= 236,194$. We can see that 90\% of skylines are discovered withing the first second and by accessing only 1\% tuples.

\section{Final Remarks}\label{sec:conclusion}
In this paper, we studied the important problem of subspace skyline discovery over datasets with categorical attributes.
We first designed a data structure for organizing tuples in candidate skyline list that supports efficient dominance check operations. We then propose two algorithms ST-S and ST-P algorithms for answering subspace skyline queries for the case where precomputed indices are absent. Finally, we considered the existence of precomputed sorted lists and developed TA-SKY,
the first threshold style algorithm for skyline discovery. 
In addition to the theoretical analysis, our comprehensive set of experiments on synthetic and real datasets confirmed the superior performance of our algorithms.

\bibliographystyle{abbrv}     
\bibliography{references}

\appendix

\section{Tree Data Structure Optimizations}\label{ap:tree-optimizations}

\vspace{1mm}
\noindent{\bf Early termination:} The tree structure described in \S\ref{subsec:tree}, does not store any information inside internal nodes. We can improve the performance of primitive operations (i.e., reduce the number of nodes visited) by storing some information inside each internal node. Specifically, each internal node maintains two variables \textit{minScore} and \textit{maxScore}. The \textit{minScore} (resp. \textit{maxScore}) value of an internal node is the \textit{minimum} (resp. \textit{maximum}) tuple score of all the tuples mapped in the subtree rooted at that node. The availability of such information at each internal 
node assists in skipping search inside irrelevant regions.

While searching the tree to discover tuples dominated by or dominating a specific tuple $t$, we also maintain an additional variable \textit{currentScore}, which initially is the same as $score(t)$ at the root of the tree. During traversals, if we follow an edge that matches the corresponding attribute value of $t$, \textit{currentScore} remains the same\footnote{An edge selected by the algorithm coming out from an internal node at level $i$ matches the attribute value of $t$ if $t[A_i] = 0$ (resp. $t[A_i] = 1$) and we follow the \textit{left} (resp. \textit{right}) edge.}. However, if the edge selected by the algorithm differs from the actual attribute value, we update the \textit{currentScore} value accordingly. In the PRUNE-DOMINATED-TUPLES($t$) operation, we compare the \textit{minScore} value of each internal node visited by the algorithm with \textit{currentScore}. If the \textit{minScore} value of a node $u$ is higher than \textit{currentScore}, we stop searching in the subtree rooted at $u$, since it's not possible to have a tuple $t'$ under $u$ that is dominated by $t$ (due to monotonicity). Similarly, while checking if $t$ is dominated by any other tuple in the tree, we stop traversing the subtree rooted at an internal node $u$ if \textit{currentScore} is higher than the \textit{maxScore} value of $u$.

Figure~\ref{fig:treeAugmented} presents the value of \textit{minScore} and \textit{maxScore} at each internal node of the tree for the relation in Table~\ref{tab:skylineTreeRunningExample}. Consider a new tuple $t = \langle 1,0,0,0 \rangle$. In order to prune the tuples dominated by $t$, we start from the root node $a$. At node $a$ \textit{currentScore} = $score(t)$ = 8. Since, $t[A_0] = 1$, we need to search both the left and right subtree of $a$. The value of \textit{currentScore} at node $c$ remains unchanged since the edge that was used to reach $c$ from $a$ matches the value of $t[A_0]$. However, for $b$  the value of \textit{currentScore} has to be updated. The \textit{currentScore} value at node $b$ is obtained by changing the value of $t[A_0]$ to 0 (values of the other attributes remain
the same as in the parent node) and compute the score of the updated tuple. Note that the value of \textit{currentScore} is less than \textit{minScore} in both nodes $b$ and $c$. Hence we can be sure that no tuple in subtrees rooted at node $b$ and $c$ can be dominated by $t$.

\begin{figure}[!h]
  \centering
  \includegraphics[scale=1.3]{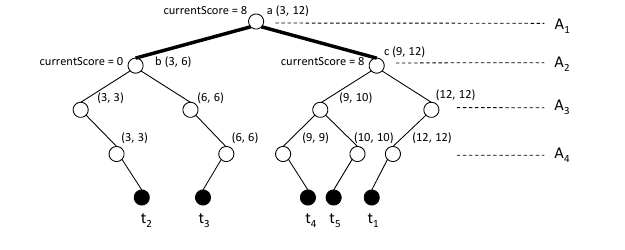}
  \vspace{-6mm}\caption{Example: Early termination}
  \label{fig:treeAugmented}
\end{figure}


\section{Extending the Data Structure for Categorical Attributes}\label{ap:STCategorical}
We now discuss how to modify ST algorithm for relations having categorical attributes. We need to make the following two changes:

\begin{itemize}
    \item The tree structure designed in \S\ref{subsec:tree} needs to be modified for categorical attribute.
    \item We also need to change the tree traversal algorithms used in each of the three primitive operations.
\end{itemize}

\noindent{\bf Tree structure:} The tree structure will not be binary anymore. In order to incorporate categorical attributes, each node $u$ at level $l$ ($1 \leq l \leq m$) of the tree now should have $|Dom(A_l)|$ children, one for each attribute value $v \in Dom(A_l)$. We shall index the edges from left to right, where the left most edge corresponds to the lowest attribute value and the attribute value corresponding to each edge increases as we move from left most edge to right most edge.

\vspace{1mm}
\noindent{\bf INSERT($t$):} After reaching a node $u$ at level $l$, select the $t[A_l]$-th child of $u$ for moving to the next level of the tree.

\noindent{\bf IS-DOMINATED($t$):} We need to follow all the edges that has index value grater or equal to $t[A_l]$.

\noindent{\bf PRUNE-DOMINATED-TUPLES($t$):} Search in all the subtrees reachable by following edges with index value less than or equal to $t[A_l]$.

%

\section{TOP-DOWN}\label{ap:top-down}

\begin{figure}[!h]
  \centering
  \includegraphics[scale=1.10]{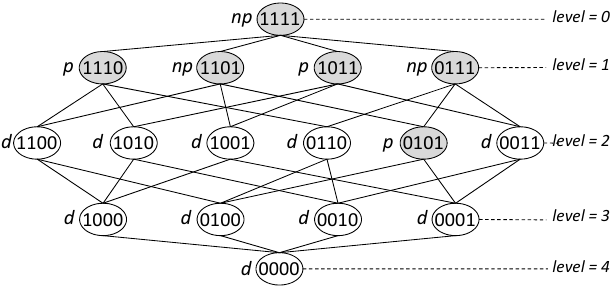}
  \caption{Nodes traversed by TOP-DOWN Algorithm}
  \label{fig:latticeTopDown}
\end{figure}

Here we provide the details of the TOP-DOWN algorithm proposed in \S~\ref{sec:topdown}.
Given a subspace skyline query $\mathcal{Q}$, consider the corresponding subspace lattice. Each node $u$ in the lattice corresponds to a {\em unique} attribute combination which can be represented by a unique $id$. We assume the existence of the following two functions, (i) $ID(C)$: returns the $id$ of an attribute value combination, and (ii) $InvID(id, m^\prime)$: returns the corresponding attribute-value combination for $id$. The details of these functions can be found in~\cite{preisinger2007hexagon}.

We observe that given a node identifier $id$, one can identify the ids of the parents (resp. children) 
of its corresponding node by calling the two functions $InvID$ and $ID$.
To do so, we first determine the corresponding attribute combination of $id$. Then identify its parents' (resp. children) 
combinations by incrementing (resp. decrementing) the value of each attribute, and finally compute the id of each combination using the function $ID$.
TOP-DOWN starts by traversing the lattice from the top node of the lattice. At this node all attributes have the maximum possible value; then conducts
 a {\em BFS} over it while constructing the level $(i-1)$ nodes from the non-empty nodes at level $i$. A node in the lattice 
is dominated if either one of its parents is dominated or there exists a tuple in the relation that matches the combination of one of its parents.

Let $id$ denote the id of the node in the lattice currently scanned by TOP-DOWN. The algorithm first identifies the parents of the current node and checks if all of them (i) have been constructed (i.e. have not been dominated) and (ii) are marked as \textit{not present} (i.e., there is no tuple in $D_{\mathcal{Q}}$ that had the combination of one of its parents).
If so, the algorithm then checks if there exist tuples in $D_{\mathcal{Q}}$ with the same attribute value combination. We use the term \textit{querying a node} in order to refer to this operation. Algorithm~\ref{alg:topDownGetTuples} presents pseudocode of this operation for a specific attribute value combination. If no such tuple exists in $D_\mathcal{Q}$, it marks $id$ as \textit{not present} and moves to the element. Otherwise, it labels $id$ as \textit{present} and outputs the tuples, returned from GET-TUPLES, as the skyline. The TOP-DOWN algorithm queries a node only when the attribute value combination corresponding to the node is incomparable with the skylines discovered earlier. The algorithm stops when there are no other ids in its processing queue.

\begin{algorithm}[htb]
\caption{{\bf GET-TUPLES}}
\begin{algorithmic}[1]
\label{alg:topDownGetTuples}
\STATE {\bf Input:} Array $values$, Sorted lists $\mathcal{L_Q}$; 
\STATE {\bf Output:} List of tuples that have the same attribute value assignment as $values$.
\STATE $tupleIDSet = \emptyset$
\STATE {\bf for} $i = 1$ to $len(values)$ {\bf do}
    \STATE \hindent[2] $currValue = values[i]$
    \STATE \hindent[2] $currtupleIDSet = $ Get all tupleIDs from $L_i \in \mathcal{L_Q}$ that has value $currValue$
    \STATE \hindent[2] $tupleIDSet = tupleIDSet \cap currtupleIDSet$

\STATE $tupleList = [\,]\,$
\STATE {\bf for} $tupleID$ in $tupleIDSet$ {\bf do}
    \STATE \hindent[2] Construct new tuple $t_{new}$ with attribute values same as $values$ and $t[tupleID] = tupleID$
    \STATE \hindent[2] $tupleList.append(t_{new})$
\STATE {\bf return} $tupleList$;
\end{algorithmic}
\end{algorithm}

The lattice structure for the subspace skyline query $\mathcal{Q}$ in Example~\ref{exmp:subspaceSkyline} is shown in Figure~\ref{fig:latticeTopDown}. Each node $u$ in the lattice represents a specific attribute value assignment in the data space corresponding to $\mathcal{Q}$. For example, the top-most node in the lattice represents a tuple $t$ with all the attribute values 1 (i.e., $t[A_i] = 1, \forall A_i \in \mathcal{Q}$). We start from the top node of the lattice. No tuple in $D_{\mathcal{Q}}$ has value 1 on all the attributes in $\mathcal{Q}$. Therefore, TOP-DOWN marks this node \textit{not present (np)}. We then move to the next level and start scanning nodes from the left. There exists a tuple $t_6 \in D_\mathcal{Q}$ with attribute values $\langle 1, 1, 1, 0 \rangle$. Hence, we mark this node \textit{present (p)} and output $t_6$ as skyline. 
The algorithm stops after querying node $\langle 0, 1, 0, 1 \rangle$. 
TOP-DOWN only needs to query 6 nodes (i.e., check 6 attribute value combinations) in order to discover the skylines. Note that the number of nodes queried by TOP-DOWN is proportional to the number of attributes in $\mathcal{Q}$ and inversely proportional to the relation size $n$. This is because with large $n$ and small $|\mathcal{Q}|$, the likelihood of having tuples in the relation that correspond to the upper-level nodes of the lattice is high.

\vspace{1mm}
\noindent{\bf Algorithm GET-TUPLES:} The algorithm to retrieve tuples in the relation matching the attribute value combination of a specific node is described in Algorithm~\ref{alg:topDownGetTuples}. The algorithm accepts two inputs: (1) $values$ array representing the value of each attribute $A_i \in \mathcal{Q}$, and (2) Sorted lists $\mathcal{L_Q}$. For each attribute $A_i \in \mathcal{Q}$ ($1 \leq i \leq m'$), the algorithm retrieves the set of tupleIDs $S_i$, that have value equals $values[i]$. This is done by performing a search operation on sorted list $L_i$. 
The set of tupleIDs that are discovered in every $S_i$ are the ids of the tuple that satisfy the current attribute value combination. We identify 
these ids by performing a set intersection operation among all the $S_i$s ($1 \leq i \leq m'$). Once the ids of all the tuples that match 
values of array $values$ are identified, the algorithm creates tuples for each id with the same attribute value and returns the tuple list.

\begin{algorithm}[!htb]
\caption{{\bf TOP-DOWN}}
\begin{algorithmic}[1]
\label{alg:st}
\STATE {\bf Input:} Query $\mathcal{Q}$, Sorted lists $\mathcal{L_Q}$; \\ {\bf Output:} $\mathcal{S}_\mathcal{Q}$.
\STATE $processed = \emptyset$;
\STATE $C = $ the attribute combination of $\mathcal{Q}$ with maximum possible value for each attribute
\STATE addQ$(queue,ID(C))$
\STATE {\bf while} $queue$ is not empty {\bf do}
    \STATE \hindent $id =$ delQ$(queue)$
    \STATE \hindent {\bf for} $pid$ in parents$(InvID(id))$
        \STATE \hindent[2] {\bf if} $pid\notin processed$ or $pid$ is marked as present
            \STATE \hindent[3] {\bf continue} {\it //skip this node}
    \STATE \hindent $tupleList = $ GET-TUPLES($values, L_\mathcal{Q}$)        
    \STATE \hindent {\bf if} $len(tupleList) == 0$:
        \STATE \hindent[2] append $processed$ by $\langle id,$\textit{not present}$\rangle$
        \STATE \hindent[2] $children = $ children$(InvID(id))$
        \STATE \hindent[2] {\bf for} $c\in children$
            \STATE \hindent[3] {\bf if} $c$ is not in $queue$: addQ$(queue,c)$
    \STATE \hindent {\bf else}:
        \STATE \hindent[2] append $processed$ by $\langle id,$\textit{present}$\rangle$
        \STATE \hindent[2] Output all the tuples in $tupleList$ as skyline.    
\end{algorithmic}
\end{algorithm}

\vspace{1mm}
\subsubsection{Performance Analysis}\label{sec:TOP-DOWN-performance}
For each non-dominated node in the lattice, the TOP-DOWN algorithm invokes the function GET-TUPLES. Hence, we measure the cost of TOP-DOWN as the number of nodes in the lattice for which we invoke GET-TUPLES, times the cost of executing GET-TUPLES function.
Since the size of all sorted lists is equal to $n$, applying binary search on the sorted lists to obtain tuples with a specified value 
on attribute $A_i$ requires $O(\log(n))$; thus the retrieval cost from all the $m'$ lists is $O(m'\log(n))$. Still taking the intersection between the lists is in $O(nm')$, which makes the worst case cost of the GET-TUPLES operation to be $O(nm')$. 
Let $k$ be the cost of GET-TUPLES operation over $\mathcal{L_Q}$, for the given relation $D$.
Moreover, considering $p_i$ as the probability that a tuple has value $1$ on the binary attribute $A_i$, we use $C(l)$ to refer to the expected cost of
 TOP-DOWN algorithm starting from a node $u$ at level $l$ of the lattice. 

\begin{theorem}\label{thm:expectedCostTopDown}
Consider a boolean relation $D$ with $n$ tuples and the probability of having value 1 on attribute $A_i$ being $p_i$, and a subspace skyline query $\mathcal{Q}$ with $m'$ attributes. The expected cost of TOP-DOWN on $D$ and $\mathcal{Q}$ starting from a node at level $l$ is described by the following recursive forumula:
\begin{small}
\begin{align}\label{eq:expectedCostTopDown}
\nonumber
C(m') &= k/m'\\
C(l) &= 
\begin{cases}
    k +  (1-p_{!\emptyset}(l)) m' C(l+1) & \quad \text{if } l = 0\\
    \frac{1}{l} \{ k + (1-p_{!\emptyset}(l) (m'-l) C(l+1) \} & \quad \text{otherwise}\\
\end{cases}
\end{align}
\end{small}
\hspace{-1mm}where $p_{!\emptyset}(l) = 1 - (1 - \prod_{i=1}^{l}(1-p_i)\prod_{i=1}^{m'-l}p_i ) ^ n$.
\end{theorem}

\begin{proof}
Consider a node $u$ at level $l$ of the lattice. Node $u$ represents a specific attribute value assignment with $l$ number of 0s and $(m' - l)$ number of 1s. Querying at node $u$ will return all tuples in dataset that have the same attribute value assignment as $u$. Let $p(t, l)$ be the probability of a tuple $t \in \mathcal{D_Q}$ having $l$ number of 0s and $(m' - l)$ number of 1s.

\begin{align}
p(t, l) = \prod\nolimits_{i=1}^{l}(1-p_i)\prod\nolimits_{i=1}^{m'-l}p_i
\end{align}

If querying at node $u$ returns at-least one tuple then we do not need to traverse the nodes dominated by $u$ anymore. However, if there exists no tuple in $D_{\mathcal{Q}}$ that corresponds to the attribute value combination of $u$, we at-least have to query the nodes that are immediately dominated by $u$. Let $p_{!\emptyset}(l)$ be the probability that there exists a tuple $t \in D_{\mathcal{Q}}$ that has the same attribute value assignment as $u$. Then,

\begin{align}
p_{!\emptyset}(l) = 1 - (1 - \prod_{i=1}^{l}(1-p_i)\prod_{i=1}^{m'-l}p_i ) ^ n
\end{align}

There are total $(m' - l)$ number of nodes immediately dominated by $u$. Therefore, Cost at node $u$ is the cost of query operation (i.e., $k$) plus with $(1 - p_{!\emptyset}(l))$ probability the cost of querying its $(m'-l)$ immediately dominated nodes.

\begin{align}
C(l) = k +  (1-p_{!\emptyset}(l)) (m' - l) C(l+1)
\end{align}

Note that a node $u$ at level $l$ has total $l$ number of immediate dominators causing the cost at node $u$ to be computed $l$ times. However, TOP-DOWN only needs to perform only one query at node $u$. Hence, the actual cost can be obtained by dividing the computed cost with value $l$.
\end{proof}

\begin{figure}[!h]
  \centering
  \includegraphics[scale=.60]{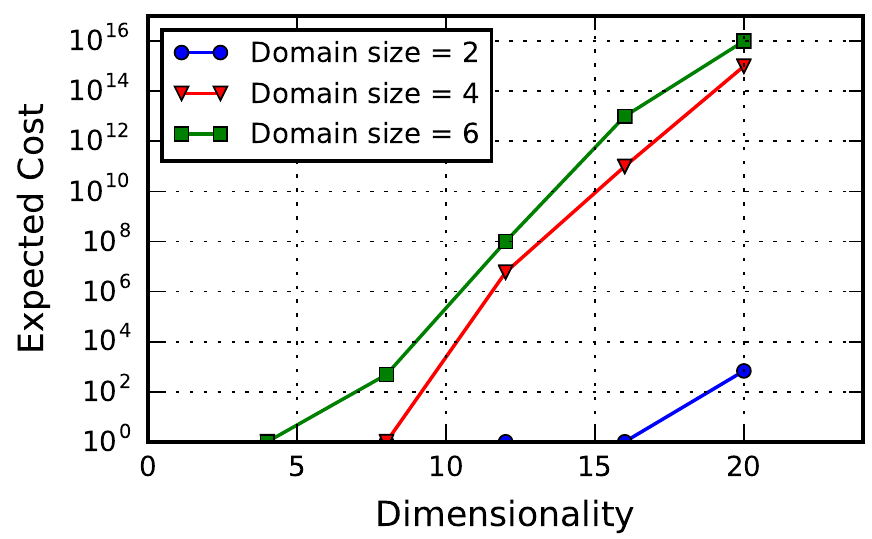}
  \vspace{-4mm}\caption{Expected number of nodes queried vs. query length}
  \label{fig:performanceTopDown}
\end{figure}

\noindent{\bf Limitation:}
We use Equation~\ref{eq:expectedCostTopDown} to compute $|C(l)|$ as a function of $|\mathcal{Q}|$ over three uniform relations containing one million tuples with cardinality 2, 4, and 6 respectively. The expected cost increases exponentially as we increase the query length. Moreover, the expected cost also increases when the attributes in $\mathcal{Q}$ have higher cardinality.

\section{Proofs}\label{sec:appendixProof}
In this section, we provide detailed proofs for the theorems from the main section of the paper.

\textsc{Theorem~\ref{thm:expectedCostSTISDominated}.} {\em
Considering a relation with $n$ binary attributes where $p_i$ is the probability that a tuple has value 1 on attribute $A_i$, the expected cost of IS-DOMINATED($t_\mathcal{Q}$) operation on a tree $T$, containing $s$ tuples is as specified in Equation~\ref{eq:expectedCostSTISDominated}.
}

\begin{proof}
Consider $t$ be the tuple for which we have to check if it is dominated. IS-DOMINATED stops the recursion when we reach a leaf node or move to a node that is empty (i.e., has no tuple mapped under it). Therefore, $C(m', s) = 1$ and $C(l, 0) = 1$.

Let us assume that we are in node $u$ at level $l$ of the tree and there are $s$ tuples mapped in the subtree rooted at $u$. 

If $t[A_l] = 0$, IS-DOMINATED first searches in the right subtree. If no tuple $t'_\mathcal{Q}$ in the right subtree dominates $t_\mathcal{Q}$, we then move to the left subtree. Let us assume the right subtree of $u$ contains $s_{right}$ number of tuples ($s_{right} \leq s$). Let $S(l, s_{right})$ be the probability that there exists a tuple in the right subtree of $u$ containing $s_{right}$ tuples that dominates $t_\mathcal{Q}$. In order for a tuple $t'_\mathcal{Q}$ to dominate $t_\mathcal{Q}$, it must have at-least value 1 on the attributes in $\mathcal{A}_{ones(t[l+1:m'])}$. This is because, since $t'[A_i] \geq t[A_i]$ ($1 \leq i \leq l-1$) and $t'[A_l] > t[A_l]$, having value 1 on attributes in $\mathcal{A}_{ones(t[l+1:m'])}$ is enough for $t'_\mathcal{Q}$ to dominate $t_\mathcal{Q}$. Hence, the probability of $t'_\mathcal{Q}$ dominating $t_\mathcal{Q}$ is $\prod\nolimits_{i=1}^{|\mathcal{A}_{ones(t[l+1:m'])}|}p_i$. Therefore,

\begin{align}
    S(l, s_{right}) = 1 - (1 - \prod\nolimits_{i=1}^{|\mathcal{A}_{ones(t[l+1:m])}|}p_i)^{s_{right}}
\end{align}

The expected cost of IS-DOMINATED, when $t[A_l] = 0$ is then,
\begin{align}
    (1- S(l, s_{right})) C(l+1, s-s_{right}) + C(l+1, s_{right})
\end{align}

If $t[A_l] = 1$, IS-DOMINATED will always search in the right subtree. Hence, the expected cost when $t[A_l] = 0$ is,
\begin{align}
    C(l+1, s_{right})
\end{align}

A node at level-$l$ containing $s$ tuples under it with the probability of having 1 on attribute $A_l$ being $p_l$, the left subtree will have $i$ tuples with the binomial probability ${s \choose i} (1-p_l)^i p_l^{s-i}$. Hence, expected cost node $u$, $C(l, s)$ is,

\begin{align}
\nonumber
1 + \sum_{i=0}^s {s \choose i} & (1-p_l)^i p_l^{s-i} ( C(l+1, s-i) + \\
                               & (1-p_l)(1-S(l, s-i))C(l+1, i))
\end{align}
\end{proof}

\textsc{Theorem~\ref{thm:expectedCostSTPruneDominatedTuples}.} {\em
Given a boolean relation $D$ with $n$ tuple and the probability of having value 1 on attribute $A_i$ being $p_i$, the expected cost of PRUNE-DOMINATED-TUPLES($t_\mathcal{Q}$) operation on a tree $T$, containing $s$ tuples is as computed in Equation~\ref{eq:expectedCostSTPruneDominatedTuples}.
}

\begin{proof}
PRUNE-DOMINATED-TUPLES($t_\mathcal{Q}$) stops the recursion when we reach a leaf node or move to a node that is empty (i.e., has no tuple mapped under it). Therefore, $C(m', s) = 1$ and $C(l, 0) = 1$.

Suppose we are in node $u$ at level $l$ of the tree and there are $s$ tuples mapped in the subtree rooted at $u$. 

If $t[A_l] = 0$, we need to search only in the left subtree. Whereas, for $t[A_l] = 1$ we need to search both the left and right subtree.

Let $p_l$ be the probability of having value $1$ on attribute $A_l$. The left subtree of node $u$ at level $l$ (with $s$ tuples under it) will have $i$ tuples with the binomial probability ${s \choose i} (1-p_l)^i p_l^{s-i}$. Hence, expected cost at node $u$, $C(l, s)$, is:

\begin{align}
\nonumber
1 + \sum_{i=0}^s {s \choose i}(1-p_l)^i & p_l^{s-i} ( (1 - p_l)C(l+1, i) +\\
                                & p_l (C(l+1, i) + C(l+1, s-i) ) )
\end{align}
\end{proof}

\textsc{Lemma~\ref{lemma:expectedDiscovery}.} {\em
Considering $p_i$ as the probability that a tuple has value 1 on the binary attribute $A_i$, the expected number of tuples discovered by TA-SKY after iterating $i$ lines is as computed in Equation~\ref{eq:expectedDiscovery}.
}
\begin{proof}
The probability that a tuple $t$ is discovered by iterating $i$ rows is one minus the probability that $t$ is not discovered in any of the $m^\prime$ lists in $\mathcal{L_Q}$. Formally:
\begin{align}
P_{seen}(t,i) = 1 - \Pi_{j=1}^{m^\prime} P_{!seen}(t,i,L_j)
\end{align}
where $P_{!seen}(t,i,L_j)$ is the probability that $t$ is not discovered at list $L_j$ until row $i$.
$P_{!seen}(t,i,L_j)$ depends on the number of $(tupleId, value)$ pairs with value $1$ in list $L_j$. A list $L_j$ has $k$ number of $(tupleId, value)$ pairs with value $1$ if the database has $k$ tuples with value $1$ on attribute $A_j$, while others have value $0$ on it. Thus, the probability that $L_j$ has $k$ number of $(tupleId, value)$ pairs with value $1$:
\begin{align}
P_{L_j}(k) = {n\choose k}(1-p_j)^{n-k}p_j^k
\end{align}
$t$ is not seen until row $i$ at list $L_j$ if either of the following cases happen:
\begin{itemize}
\item $t[A_j] = 0$ and (considering the random positioning of tuples in lists) $t$ is located after position $i$ in list $L_j$ for all the cases that $L_j$ has $k \, (k < i)$ number of $(tupleId, value)$ pairs with value $1$. 
\item $t[A_j] = 1$ and (considering the random positioning of tuples in lists) $t$ is located after position $i$ in list $L_j$ for all the cases that $L_j$ has $k \, (k > i)$ number of $(tupleId, value)$ pairs with value $1$.
\end{itemize}
Thus:
\begin{align}
\nonumber
&P_{!seen}(t,i,L_j) = \\ 
&(1 - p_j) \Big(\sum_{k=0}^{i-1}P_{L_j}(k)\frac{n-i}{n-k} + \sum_{k=i}^{n}P_{L_j} \Big) + p_j\sum_{k=i+1}^n P_{L_j}(k) \frac{k-i}{k}
\end{align}
We now can compute $P_{seen}(t,i)$ as following:
\begin{align}
\nonumber
&P_{seen}(t,i) = \\
&1 - \prod_{j=1}^{m^\prime} \bigg( (1 - p_j) \Big(\sum_{k=0}^{i-1}P_{L_j}(k)\frac{n-i}{n-k} + \sum_{k=i}^{n}P_{L_j} \Big) +\\ \nonumber &p_j\sum_{k=i+1}^n P_{L_j}(k) \frac{k-i}{k} \bigg)
\end{align}
Having the probability of a tuple being discovered by iterating $i$ lines, the expected number of tuples discovered by iterating $i$ lines is:
\begin{align}
\nonumber
&E_{seen}[i] = n P_{seen}(t,i) = \mbox{ Equation~\ref{eq:expectedDiscovery}}
\end{align}
\end{proof}

\textsc{Theorem~\ref{thm:expectedCostTA-SKY}.} {\em
Given a subspace skyline query $\mathcal{Q}$, the expected number of sorted access performed by TA-SKY on a $n$ tuple boolean database with probability of having value $1$ on attribute $A_j$ being $p_j$ is,
\begin{align*}
m^\prime \sum_{i=1}^n i\times P_{stop}(i)
\end{align*}
where $P_{stop}(i)$ is computed using Equations~\ref{eq:stopi-1},~\ref{eq:stopi-2}, and~\ref{eq:stopi-3}.
}

\begin{proof}
Let us first compute the probability that algorithm stops after visiting $i$ rows of the lists. Please note that the algorithm checks the stopping condition at iteration $i$ if $cv_{ij} = 0$ for at least one sorted list. Thus the algorithm stops when (1) $cv_{ij} = 0$ for at least one sorted list AND (2) there exists a tuple among the discovered ones that dominates the maximum possible tuple in the remaining lists.

Suppose $i^\prime$ tuples have seen at least in one of the list so far. Using Lemma~\ref{lemma:expectedDiscovery} we can set $i^\prime = E_{seen}[i]$. Let $P_{j0}(i)$ be the probability that $cv_{ij} = 0$ for sorted list $L_j$.
\begin{align}
P_{j0} = (1 - p_j)^{n-i}
\end{align}

Moreover, Consider $P_0(i, k)$ be the probability that after iteration $i$, $cv_i = 0$ for $k$ sorted lists and $\mathcal{Q}_k$ is corresponding attribute set. Therefore,
\begin{align}
P_0(i, k) = {m^\prime \choose k} \prod_{A_j \in \mathcal{Q}_k} P_{j0} \prod_{A_j \in \mathcal{Q} \setminus \mathcal{Q}_k} (1 - P_{j0})
\end{align}

For a given setting that $cv_i = 0$ for $k$ sorted lists, the algorithm stops, {\it iff} there exists at least one tuple among the discovered ones that dominate the maximum possible value in $m'$ sorted lists; i.e. the value combination that has $0$ in $k$ and $1$ in all the remaining $m^\prime - k$ positions.

A tuple $t$ need to have the value $1$ in all the $m^\prime - k$ list and also \emph{at least one value $1$ in one of the $k$ lists ($\mathcal{Q}_k$)} to dominate the maximum possible remaining value. The probability that a given tuple satisfies this condition is:
\begin{align}
P_{stop}(t, \mathcal{Q}_k) = \underset{\forall A_j \in \mathcal{Q}\backslash \mathcal{Q}_k} {\Pi} p_j (1-\underset{\forall A_j \in \mathcal{Q}_k} {\Pi} (1 - p_j))
\end{align}
Thus, the probability of having at least one tuple that satisfies the dominating condition is:
\begin{align}
P_{dominate}(i, k) = {m^\prime \choose k} \times (1 - (1 - P_{stop}(t, \mathcal{Q}_k))^{i^\prime})
\end{align}

We now can compute the probability distribution of the algorithm cost as following:
\begin{align}
P_{stop}(i) = \sum_{k=1}^m P_0(i, k) \times P_{dominate}(k)
\end{align}

Finally, the expected number of sorted access performed by TA-SKY is:
\begin{align}
m^\prime \sum_{i=1}^n i\times P_{stop}(i)
\end{align}
\end{proof}       

\end{document}